\documentclass{article}
\usepackage{amssymb,latexsym,amsfonts,amsmath,mathrsfs,amsthm,graphicx}
\usepackage{geometry}

\usepackage{setspace}
\usepackage{hyperref}
\usepackage{subfigure}

\hypersetup{colorlinks=true, linkcolor=blue, citecolor=red}
\usepackage{appendix}
\usepackage{bbm}
\usepackage{dsfont}
\newcommand{\p}{^{\prime}}
\newtheorem{theorem}{Theorem}[section]
\newtheorem{proposition}{Proposition}[section]
\newtheorem{lemma}{Lemma}

\newcommand{\xqedhere}[2]{%
  \rlap{\hbox to#1{\hfil\llap{\ensuremath{#2}}}}}

\allowdisplaybreaks

\title{A Multi-Strain Virus Model with Infected Cell Age Structure:  Application to HIV}
\date{}
\author{Cameron J. Browne\footnote{E-mail: {\tt cameron.j.browne@vanderbilt.edu}.
Department of Mathematics, Vanderbilt University, Nashville, TN.}}
\begin{document}
\maketitle

\begin{abstract}
A general mathematical model of a within-host viral infection with $n$ virus strains and explicit age-since-infection structure for infected cells is considered.  In the model, multiple virus strains compete for a population of target cells.  Cells infected with virus strain $i\in\left\{1,...,n\right\}$ die at per-capita rate $\delta_i(a)$ and produce virions at per-capita rate $p_i(a)$, where $\delta_i(a)$ and $p_i(a)$ are functions of the age-since-infection of the cell.   Viral strain $i$ has a basic reproduction number, $\mathcal{R}_i$, and a corresponding positive single strain equilibrium, $E_i$, when $\mathcal{R}_i>1$.  If $\mathcal{R}_i<1$, then the total concentration of virus strain $i$ will converge to $0$ asymptotically.  The main result is that when $\max_i \mathcal{R}_i>1$ and all of the reproduction numbers are distinct, i.e. $\mathcal{R}_i\neq \mathcal{R}_j \ \forall i\neq j$, the viral strain with the maximal basic reproduction number competitively excludes the other strains.  As an application of the model, HIV evolution is considered and simulations are provided.  \\ \\
\noindent {\bf Keywords:} mathematical model, virus dynamics, age-structure, global stability analysis, multi-strain, competitive exclusion, Lyapunov functional, infinite-dimensional dynamical system
\end{abstract}

\section{Introduction}

Mathematical modeling of within-host virus dynamics has been an extensive subject of research over the past two decades.  Many of the models have been related to a differential equation system introduced by Perelson et al. in 1996 \cite{perelson}, often referred to as the standard virus model.  The standard model describes the coupled changes in target cells, infected cells, and free virus particles through time in a single compartment of an infected individual.  The model has been very useful in quantifying certain parameters, especially for HIV, and providing insights for viral infections.

De Leenheer and Smith rigorously characterized the dynamical properties of the standard virus model \cite{deleenheerandsmith}.  They found that a quantity known as the basic reproduction number, $\mathcal{R}_0$, largely determines the global dynamics of the system.  If $\mathcal{R}_0<1$, then the virus is cleared.  On the other hand, when $\mathcal{R}_0>1$, a unique positive equilibrium exists, but oscillatory behavior can not be ruled out in general.  De Leenheer and Pilyugin found a sufficient condition for global stability of the positive equilibrium by placing restrictions on the net natural growth rate of the uninfected cell population (which they called the ``sector condition'') and utilizing a Lyapunov function \cite{deleenheerandpilyugin}.

However, the standard model does not include many relevant factors present in within-host virus dynamics.  The standard virus model assumes simultaneous infection of target cells and viral production, and hence ignores intracellular delays.  To account for the time lag between viral entry of a target cell and subsequent viral production from the newly infected cell, Perelson et al. included discrete and distributed delays in the standard model \cite{delaymod}.  Nelson et al. considered a model with age structure in the infected cell component, which generalizes the delay standard virus model by allowing for infected cell death rate and viral production to vary with age since infection of an infected cell \cite{agemodel0}.  This model has appeared often in the literature \cite{agemodel1, agemodel2,agemodelGAS,agemodelrong} and the global dynamics were analyzed in \cite{bropil}.

In addition, multiple strains or populations of viruses often occur in one host as a result of within-host evolution or several infection events.  The question then arises; what are the fate of multiple virus strains or species competing for the same target cell population?  De Leenheer and Pilyugin studied a multi-strain version of the standard model with and without mutations, and established competitive exclusion when mutations are not present and the sector condition holds \cite{deleenheerandpilyugin}.  Introducing small mutation rates produces multi-strain persistence, but simply perturbs the viral steady states of the no-mutation model.  Many other studies have investigated multiple strains in within-host virus models \cite{althaus, ball, bonhoeffer, malaria, rongaverage,souza}.

In this paper, we present a global analysis of a within-host virus model with, both, multiple virus strains and age structure in the infected cell compartments of the various strains.  Cells infected with virus strain $i$ die at per-capita rate $\delta_i(a)$ and produce virions at per-capita rate $p_i(a)$, where both rates are functions of the age-since-infection of the cell.  Thus, we allow for each viral strain to have a distinct infected cell life history and compete for a common target cell population.  Incorporating non-constant viral production rates and infected cell death rates when investigating the evolution of viruses and the dynamics of strain replacement has been of recent interest  \cite{agemodel1, althaus}.   Our main result is that the competitive exclusion principle and the principle of $\mathcal{R}_0$ maximization hold in this model, i.e. the system will converge to a steady state where the virus strain with maximal reproduction number persists and all other viral strains are extinct.  The global analysis required for this proof is complicated by the fact that the underlying state space for an age-structured model is infinite dimensional.

Recently, there has been progress in the global analysis of infection-age structured models via Lyapunov functionals.  McCluskey and others have incorporated an integration term into a Lyapunov functional form often utilized for Lotka-Volterra type ODE models \cite{bropil,magal, mcc}.  The application of the Lyapunov functional in age-structured models requires more delicate analysis than the case of ODEs. This often entails proving asymptotic smoothness of the semigroup generated by the family of solutions and proving existence of an interior global attractor, and then defining a Lyapunov functional on this attractor.  In this paper, we modify this approach in order to maximize the utility of the Lyapunov functional that we found for our system.  We still need to prove existence of an interior global attractor, but we can employ strong mathematical induction and utilize the Lyapunov functional in order to establish uniform persistence, from which existence of an interior global attractor follows.  

The paper is organized as follows:  In Section \ref{s2}, we introduce a general formulation of the model.   In Section \ref{s3}, we show existence of $C^0$ semigroup generated by solutions to the model and prove some important properties of the semigroup.  In section \ref{s4}, we define the reproduction number, $\mathcal{R}_i$, of each strain, and prove that a virus strain is cleared if its reproduction number is less than unity.   In Section \ref{s5}, we prove the main result that competitive exclusion occurs.  In Section \ref{s6}, numerical simulations illustrate the result and we provide insight into the transient dynamics with application to HIV evolution.  In Section \ref{discuss}, we provide a discussion of the results and outline future work. 

\section{Model Formulation}\label{s2}

We extend the standard virus model by considering multiple virus strains and allowing for infected cell death rate and viral production to vary with age since infection of an infected cell.  Consider the following model:

\begin{align*}
\frac{dT(t)}{dt} &= f(T(t))- \sum_{i=1}^n k_iV_i(t)T(t), \\
 \frac{dV_i(t)}{dt} &= \int_0^{\infty} \!p_i(a)T_i^{*}(t,a)\,da - \gamma_i V_i(t),  \quad i=1,...,n \\
\frac{ \partial  T_i^{*}(t,a)}{\partial t}+\frac{\partial  T_i^{*}(t,a)}{\partial a} &=  - \delta_i(a) T_i^{*}(t,a), \quad i=1,...,n \label{a1} \tag{1} \\
T_i^*(t,0)&=k_iV_i(t)T(t),
\end{align*}
where $T(t)$ is the concentration of uninfected cells and $V_i(t)$ is the concentration of free virus particles of strain $i$.   $T^*_i(t,a)$ denotes the density, with respect to age since infection, of infected cells which are infected by virus strain $i$.  
  
The function $f(T)$ represents the net growth rate of the uninfected cell population.  The parameters $k_i$ and $\gamma_i$ are the infection rate and clearance rate for virus strain $i$, respectively.  The net growth rate $f(T)$ is assumed to be smooth and satisfy the following property:  there exists $\overline{T}_0>0$ such that:
\begin{align*}
f(T)>0 \ \ \text{for all } 0\leq T < \overline{T}_0, \quad  \text{and } f(T)<0 \ \ \text{for all } T>\overline{T}_0. \tag{2}\label{2.1.0}
\end{align*}
By continuity of $f$, $f(\overline{T}_0)=0$.  Thus, $\overline{T}_0$ is the equilibrium concentration of target cells in an uninfected individual.  Two commonly used functional forms for $f(T)$ are:
\begin{enumerate}
\item $f(T)=f_1(T)=s-cT$   (Nowak and May) \cite{nowak}
\item $f(T)=f_2(T)=s-cT+rT(1-T/T_{max})$   (Perelson and Nelson) \cite{stmodel}
\end{enumerate}
Both $f_1(T)$ and $f_2(T)$ satisfy Condition (\ref{2.1.0}).  The first form, $f_1(T)$, is a simple linear function, which assumes that cells are supplied at a constant rate $s$ from a source such as the thymus, and die at the (per-capita) rate $c$.  $f_2(T)$ adds a logistic proliferation term to the equation.

The functions $\delta_i(a)$ and $p_i(a)$ are the infection-age dependent (per-capita) rates of infected cell death and virion production for infected cells infected with virus strain $i$, respectively.  The functions $\delta_i(a)$ and $p_i(a)$ are assumed to be in $L^{\infty}_+$, the non-negative cone of $L^{\infty}(0,\infty)$.  Let $\kappa>0$ be an upper bound for the functions $p_i(a)$, i.e. $p_i(a)\leq\kappa \ \text{a.e.} \ \forall i$.  We suppose further that $\exists b>0$ such that $\delta_i(a)\geq b \ \forall i$ a.e. on $[0,\infty)$.  

There are multiple simplifying assumptions in the model (\ref{a1}).  First, the terms $-k_iV_iT$ associated with the loss of free virus particles due to absorption in target cell upon infection have been ignored in the $\frac{dV_i}{dt}$ equations.  This is a common assumption in HIV models since the loss terms are considered relatively small and can be absorbed into the virus clearance rates $\gamma_i$ \cite{stmodel}.  Another assumption we make is that viruses of different strains cannot infect the same cell.  In reality for HIV, cells can become infected by multiple virus strains, although co-infected cells represent a small fraction of infected cells \cite{althaus}.  Allowing for co-infection or super-infection of cells would add significant complexity to the model (\ref{a1}) and the analysis, hence we leave this for future studies.  

Various approaches have been developed for analyzing age structured models.  The general idea is to study the nonlinear semigroup generated by the family of solutions.  One approach is to use the theory of integrated semigroups \cite{magal,thieme}.  We employ another method, namely integrating solutions along the characteristics to obtain an equivalent integro-differential equation.  This approach was utilized by Webb for age-dependent population models \cite{webb}.

For $i=1,...,n$, define
\begin{align*}
\phi_i(a)=e^{-\int_0^a \delta_i(s)\,ds}. \tag{3}\label{5.2.4}
\end{align*}
The function $\phi_i(a)$ can be interpreted as the probability that an infected cell (infected with strain $i$) will survive to age $a$.
Then, integrating along the characteristics, we arrive at the following more general formulation:
\begin{align*}
\frac{dT(t)}{dt} &= f(T(t))- \sum_{i=1}^n k_iV_i(t)T(t), \\
 \frac{dV_i(t)}{dt} &= \int_0^{\infty} \!p_i(a)T_i^*(t,a)\,da- \gamma V_i(t), \label{a3} \tag{4} \\
T_i^*(t,a) &= \phi_i(a) k_iV_i(t-a)T(t-a) \mathds{1}_{\left\{t>a\right\}} +\frac{\phi_i(a)}{\phi_i(a-t)}T_i^*(0,a-t)\mathds{1}_{\left\{a>t\right\}}\\
T_i^*(t,0)&=k_iV_i(t)T(t) \quad  T_i^*(0,a) \in L^1_+(0,\infty), \\
T(0)& \in  \mathbb{R}_+, \quad V_i(0) \in \mathbb{R}_+,
\end{align*}
where $L^1_+(0,\infty)$ is the non-negative cone of $L^1(0,\infty)$, $\mathbb{R}_+=[0,\infty)$ and $\mathds{1}_{\left\{t>a\right\}}$ is the indicator function for the set $\left\{a\in (0,\infty): t>a\right\}$.  Define the state space $X$ as 
$$X=\mathbb{R}^{n+1}_+\times \prod_1^n L^1_+,$$
where $L^1_+=L^1_+(0,\infty)$ and $\mathbb{R}^{n+1}_+$ is the non-negative orthant of $\mathbb{R}^{n+1}$.  Note that $X$ is a closed subset of a Banach Space, and hence is a complete metric space.  The norm on $X$ is taken to be:
$$\left\|x\right\|=|T| + |V_1|+....+|V_n|+\int_0^{\infty}|T^*_1(a)|\,da+......+\int_0^{\infty}|T^*_n(a)|\,da$$
for $x=\left(T,V_1,....,V_n,T^*_1(a),.....,T^*_n(a)\right)\in X$.  Hence, the norm represents the total concentration of the healthy cells, infected cells, and virus in the body.

\section{Existence and properties of semigroup}\label{s3}
\subsection{Existence and boundedness}
The local existence, uniqueness, and non-negativeness of solutions to the system (\ref{a3}) can be demonstrated.
\begin{proposition}\label{exist}
 Let $x_0\in X$.  For any neighborhood $B_0\subset  X$ with $x_0\in B_0$, there exists an $\epsilon>0$ and a unique continuous function, $\psi:[0,\epsilon]\times B_0 \rightarrow X$ where $\psi(t,x)$ is the solution to the model (\ref{a3}) with $\psi(0,x)=x$.
\end{proposition}
\begin{proof}
Existence and uniqueness can be proved by formulating the solution to the system (\ref{a3}) as a fixed point of an integral operator, $\Lambda$, on an appropriate closed subset of $C\left([0,\epsilon]\times B_0,\widehat{X}\right)$, the set of continuous functions from $[0,\epsilon]\times B_0$ to $\widehat{X}$, where $\widehat{X}:=\mathbb{R}^{n+1}\times \prod_1^n L^1(0,\infty)$.  For $\epsilon>0$ sufficiently small, this map is a contraction, and hence, by the contraction mapping theorem, we obtain local existence and uniqueness of solutions to the system (\ref{a3}) (in the larger state space $\widehat{X}$).  Then, we define the transformations $\widetilde{V}_i(t)=e^{\gamma_i t}V_i(t)$, and show with a similar contraction argument that the transformed system has a unique solution whose state variables remain in the state space $X$, implying non-negativeness of the original solution.  The details are contained in \cite{bropil}, where the theorem is proved for the single-strain model (the case $n=1$).
\end{proof}
  Note that solutions to the system (\ref{a3}) are solutions to the system (\ref{a1}) if they have appropriate differentiability in the variable $a$.  If not, solutions to the system (\ref{a3}) are weak solutions to the system (\ref{a1}).  
  
 Next, we establish existence of a semigroup $S(t)$ generated by solutions to the model (\ref{a3}) and find that $S(t)$ is point dissipative.
\begin{proposition}\label{bounded}
Solutions to the system (\ref{a3}) remain bounded in forward time .  Therefore, the family of solutions to the system (\ref{a3}) form a $C^0$ semigroup on $X$, which we call $S(t)$.  Moreover, the semigroup $S(t)$ is point dissipative, i.e. there exists a bounded set $B\subset X$ which attracts all points in $X$ $\left( \forall x\in X, \ d\left(S(t)x,B\right)\rightarrow 0 \ \text{as} \ t\rightarrow \infty\right)$.
\end{proposition}
\begin{proof}
If solutions can be shown to remain bounded in forward time, then existence of the semigroup can be established.  Indeed, for $t\geq 0$ define the flow $S(t):X \rightarrow X$ as $S(t)x=\psi(t,x)$, where $\psi(t,x)$ is the solution to the model (\ref{a3}) with initial condition $x$.  The family of functions $\left\{S(t)\right\}_{t\geq 0}$ satisfy the properties of a $C^0$ semigroup on $X$ \cite{hale1} (the semigroup property and continuity are a consequence of Proposition \ref{exist}).  Boundedness in forward time and point dissipativity (assuming boundedness in forward time) can be proved with the same argument (this will become apparent in the next paragraph).  Hence, we suppose that the solutions are forward complete, i.e. exist on the time interval $[0,\infty)$, and show that $S(t)$ is point dissipative.  

By looking at the integral equations, we observe that $T(t)$, $V_i(t)$, and $\int_0^\infty \! T^*_i(t,a)\,da$ are differentiable in $t$ (for all $i=1,...,n$) by the fundamental theorem of calculus for $T(t), V_i(t)$ and for the case of $\int_0^{\infty}\!T^*_i(t,a) \,da$, the smoothing properties of convolution.  Also, the assumption on $f(T)$ imply there exists $A>0$ and $B>0$ such that $f(T)\leq A-BT$.  Let $\gamma=\min(\gamma_1,....,\gamma_n)$ and consider $T+\sum_{i=1}^n{\int_0^{\infty}\! T^*_i \,da} + \frac{b}{2\kappa}\sum_{i=1}^nV_i$.  Integrating over all ages $a$ in the partial differential equation in the model (\ref{a1}) and adding time derivatives of the model components, we obtain:
\begin{align*}
\frac{d}{dt}\left(T+\sum_{i=1}^n{\int_0^{\infty}\! T^*_i \,da} + \frac{b}{2\kappa}\sum_{i=1}^nV_i\right)&=f(T)-\sum_{i=1}^n k_iV_iT+\sum_{i=1}^n\left[k_iV_iT-\lim_{a\rightarrow\infty}T^*_i(t,a) \right. \\
& \qquad \qquad \left. -\int_0^{\infty}\!\delta_i(a)T^*\,da + \frac{b}{2\kappa} \left( \int_0^{\infty}\! p_i(a)T^*_i \,da -\gamma_i V_i\right)\right] \\
&\leq A-BT-b\sum_{i=1}^n\int_0^{\infty}\! T^*_i\,da + \frac{b}{2\kappa}\kappa\int_0^{\infty}\! T^*_i\,da - \frac{b}{2\kappa}\gamma \sum_{i=1}^n V_i \\
&= A-BT -\frac{b}{2} \sum_{i=1}^n\int_0^{\infty}\! T^*_i\,da -\frac{b}{2\kappa}\gamma  \sum_{i=1}^n V_i \\
&\leq A - \alpha \left(T+ \sum_{i=1}^n\int_0^{\infty}\! T^*_i \,da + \frac{b}{2\kappa} \sum_{i=1}^n V_i\right)
\end{align*}
where $\alpha=\min(B,\frac{b}{2},\gamma)$.  This implies that $\limsup_{t\rightarrow \infty}(T+\int_0^{\infty}\! T^* \,da + \frac{b}{2\kappa}V)\leq \frac{A}{\alpha}$.  Hence, the semigroup $S(t)$ is point dissipative.
\end{proof}

\subsection{Asymptotic smoothness}
Next, we establish asymptotic smoothness of the semigroup.  The semigroup $S(t)$ is \emph{asymptotically smooth}, if, for any nonempty, closed bounded set $B\subset X$ for which $S(t)B\subset B$, there is a compact set $J\subset B$ such that $J$ attracts $B$.    A definition which is useful in proving asymptotic smoothness is the following:  The semigroup $S(t)$ is \emph{completely continuous} if for each $t>0$ and each bounded set $B\subset X$, we have $\left\{S(s)B, 0\leq s \leq t\right\}$ is bounded and $S(t)B$ precompact.  We will apply the following theorem:
\begin{theorem}[\cite{hale1}]\label{asymsmooth}
For each $t\geq 0$, suppose $S(t)=U(t)+C(t):X\rightarrow X$ has the property that $C(t)$ is completely continuous and there is a continuous function $k:\mathbb{R}^+\times \mathbb{R}^+\rightarrow \mathbb{R}^+$ such that $k(t,r)\rightarrow 0$ as $t\rightarrow\infty$ and $\left\|U(t)x\right\|\leq k(t,r)$ if $\left\|x\right\| \leq r$.  Then $S(t), t\geq 0$, is asymptotically smooth.
\end{theorem}
Since $L^1_+$ is a component of our state space $X$, we need a notion of compactness in $L^1_+$.   Being an infinite dimensional space, boundedness does not imply precompactness.  We use the following result.
\begin{theorem}[\cite{adams}]\label{compact}
Let $K\subset L^p_+(0,\infty)$ be closed and bounded where $p\geq 1$.  Then $K$ is compact iff the following hold:
\begin{itemize}
\item[(i)] $\lim_{h\rightarrow 0}\int_0^{\infty}\! |u(z+h)-u(z)|^p \, dz=0$ uniformly for $u\in K$.  ($u(z+h)=0$ if $z+h<0$).
\item[(ii)] $\lim_{h\rightarrow\infty}\int_h^{\infty}\!|u(z)|^p \, dz =0$ uniformly for $u\in K$.
\end{itemize}
\end{theorem}
Using this $L^p$ compactness condition and Theorem \ref{asymsmooth}, we can establish the following proposition.
\begin{proposition}\label{smooth}
The semigroup $S(t)$ is asymptotically smooth.
\end{proposition}
\begin{proof}
Suppose that $B\subset X$ is bounded with $\sup_{x\in B}\left\|x\right\|\leq r$.  Define the projection of $S(t)B$ on to $\mathbb{R}^{n+2}$ as $\pi^0 S(t)B$.  Then $\pi^0 S(t)B$ is precompact because solutions remain bounded.  Now define the the projection of the semigroup $S(t)$ on to the $T^*_i(t,a)$ component in $L^1_+$ as $\pi_i S(t)$.  We will show that  $\pi_iS(t)=U_i(t)+C_i(t)$, where there exists $k(t,r)\rightarrow 0$ as $t\rightarrow\infty$ with $\left\|U_i(t)x\right\|\leq k(t,r)$ if $\left\|x\right\| \leq r$, and for any $B\subset X$ which is closed and bounded, we have $C_i(t)B$ is compact.  Then we can apply Theorem \ref{asymsmooth} for $S(t)=U(t)+C(t)$ where
$$U(t)= \left(0,....,0,U_1(t),.....,U_n(t) \right), \qquad C(t)= \left( \pi^0 S(t),C_1(t),....,C_n(t)\right) $$
Indeed, if $B\subset X$ is closed and bounded, then $C(t)B \subset \pi^0S(t)B \times \prod_{i=1}^n C_i(t)B$ is a closed subset of a compact set, and hence is compact.  Also, the decaying requirement for $U(t)$ is certainly satisfied.
In order to follow this plan of action, let $\pi_iS(t)= U_i(t)+C_i(t)$ where
\begin{align*}
(U_i(t)x)(a)&= \frac{\phi_i(a)}{\phi_i(a-t)}T^*_i(0,a-t)\mathds{1}_{\left\{a>t\right\}} ,\\
(C_i(t)x)(a)&= \phi_i(a)k_iV_i(t-a)T(t-a)\mathds{1}_{\left\{a<t\right\}}
\end{align*}
 Then
\begin{equation*}
\left\| U_i(t)x\right\| = \int_t^{\infty}\! \frac{\phi_i(a)}{\phi_i(a-t)}T^*_i(0,a-t) \, da  \leq e^{-bt}  \int_t^{\infty}\!T^*_i(0,a-t) \, da  \leq e^{-bt} \left\| T^*_i(0,\cdot) \right\| .
\end{equation*}
Hence, if we let $k(t,r)=re^{-bt}$, then certainly $k(t,r)\rightarrow 0$ as $t\rightarrow\infty$ and $\left\|U_i(t)x\right\|\leq k_i(t,r)$ if $\left\|x\right\| \leq r$.
To show that $C_i(t)$ satisfies the compactness condition, we apply Theorem \ref{compact}.

Let $B\subset X$ be closed and bounded.  Suppose $r>0$ such that $\left\|x\right\| \leq r$ for all $x\in B$.  Notice that for all $x\in B$, $\int_h^{\infty} \! |(C_i(t)x)(a)| \,da =0 \ \ \forall h\geq t$.  Therefore (ii) is satisfied for the set $C_i(t)B\subset L^1$.
To check condition (i), observe:
\begin{align*}
\int_0^{\infty}\! | (C_i(t)x)(a)&-(C_i(t)x)(a+h)|\,da  \\
&= \int_0^t \! \left| \phi_i(a)k_iV_i(t-a)T(t-a) -  \phi_i(a+h)k_iV_i(t-a-h)T(t-a-h)\right| \,da \\
                                                                            &=  \int_0^t \!  \phi_i(a)\left|k_iV_i(t-a)T(t-a) -  \frac{\phi_i(a+h)}{\phi_i(a)}k_iV_i(t-a-h)T(t-a-h)\right| \,da \\
                                                                           &\leq \int_0^t \! e^{-ba}\left| k_iV_i(t-a)T(t-a)-\frac{\phi_i(a+h)}{\phi_i(a)}k_iV_i(t-a)T(t-a) \right| \,da  \\
& \  \ +   \int_0^t \! e^{-ba}\frac{\phi_i(a+h)}{\phi_i(a)}\left| k_iV_i(t-a)T(t-a) - k_iV_i(t-a-h)T(t-a-h) \right| \,da  \tag{6} \label{a6}
\end{align*}
Let $M=\max(r,\frac{2\kappa A}{b\alpha})$ where $A,\alpha$ are defined in Proposition \ref{bounded}.  Notice that
\begin{align*}
 \int_0^t \!& e^{-ba}\left| k_iV_i(t-a)T(t-a)-\frac{\phi_i(a+h)}{\phi_i(a)}k_iV_i(t-a)T(t-a) \right| \,da  \\
                                         &= \int_0^t \! e^{-ba}k_iV_i(t-a)T(t-a)\left(1-\frac{\phi_i(a+h)}{\phi_i(a)}\right) \,da \\
													&\leq M \int_0^{\infty}\! e^{-ba}\left(1-\frac{\phi_i(a+h)}{\phi_i(a)}\right) \,da ,
\end{align*}
\begin{align*}
\lim_{h\rightarrow 0} \int_0^{\infty}\! e^{-ba}\left(1-\frac{\phi_i(a+h)}{\phi_i(a)}\right) \,da &= \int_0^{\infty}\! e^{-ba}\left(1-\lim_{h\rightarrow 0}\frac{\phi_i(a+h)}{\phi_i(a)}\right) \,da =0,
\end{align*}
where we applied Dominated Convergence Theorem.  Also,
\begin{align*}
& \int_0^t \! e^{-ba}\left| \frac{\phi_i(a+h)}{\phi_i(a)}k_iV_i(t-a)T(t-a) -  \frac{\phi_i(a+h)}{\phi_i(a)}k_iV_i(t-a-h)T(t-a-h) \right| \,da \\
& \ \ \  \leq k \sup_{\tau\in [0,t]} | V_i(\tau)T(\tau)-V_i(\tau-h)T(\tau-h)| \int_0^{\infty}\! e^{-ba} \,da  \\
&  \ \leq k \sup_{\tau\in [0,t]} \left( |V_i(\tau)|\cdot |T(\tau)-T(\tau-h)| + |T(\tau-h)|\cdot |V_i(\tau)-V_i(\tau-h)| \right) \int_0^{\infty}\! e^{-ba} \,da \tag{7} \label{a7}
\end{align*}
By the integral formulation, we find that
\begin{align*}
|V_i(\tau)-V_i(\tau-h)|&=\left|\int_{\tau-h}^{\tau}\!\int_0^{\infty}\!p_i(a)T^*_i(s,a)\,da\,ds - \gamma_i\int_{\tau-h}^{\tau}\!V_i(s)\,ds \right|  \\
&\leq h(\kappa \left\|T^*_i\right\| +\gamma_i \left\|V_i\right\|) \\
&\leq h(\kappa +\gamma_i)r
\end{align*}
\begin{align*}
 |T(\tau)-T(\tau-h)| &\leq \int_{\tau - h}^{\tau} \! |f(T(s))-k_iV_i(s)T(s)| \,ds \\
&\leq \left(\max_{s\in [0,r]}|f(s)|+r^2 \right)h
\end{align*}
  Hence, by Inequality \ref{a7},
\begin{align*}
 \int_0^t \! e^{-ba}\left| \frac{\phi_i(a+h)}{\phi_i(a)}k_iV_i(t-a)T(t-a) -  \frac{\phi_i(a+h)}{\phi_i(a)}k_iV_i(t-a-h)T(t-a-h) \right| \,da  &\leq h\cdot M
\end{align*}
where $M=r(\kappa +\gamma_i + \max_{s\in [0,r]}|f(s)|+r^2) \int_0^{\infty}\! e^{-ba} \,da$.  This converges uniformly to $0$ as $h\rightarrow 0$.  Therefore the equation (\ref{a6}) converges uniformly to $0$ as $h\rightarrow 0$ and condition (i) is proved for $C_i(t)B$.  Hence, by Theorem \ref{compact}, $C_i(t)B$ is compact.  By the aforementioned argument we can apply Theorem \ref{asymsmooth} and conclude that $S(t)$ is asymptotically smooth.
\end{proof}

\subsection{Limit sets and attractor}\label{prelim}

In this subsection, we recall several definitions concerning semigroup dynamics in infinite dimension.  We also prove two simple propositions about limit sets that will applied later in our analysis and state a theorem about existence of a global attractor.  

A positive orbit exists for all $x\in X$, however, a negative orbit need not exist for all $x\in X$ since the semigroup $S(t)$ is not onto.  When a negative orbit does exist for a point $x$, then we can find a complete orbit through $x$.  A \emph{complete orbit} through $x$ is a function $z:\mathbb{R}\rightarrow X$ such that $z(0)=x$ and, for any $s\in\mathbb{R}$, $S(t)z(s)=z(t+s)$ for $t\geq 0$.  
The \emph{omega limit set} of $x$, $\omega(x)$, is defined as
$$\omega(x):= \left\{ y\in X: \exists \ t_n \uparrow \infty \ \ \text{such that  } S(t_n)x\rightarrow y \right\}.$$
The \emph{alpha limit set corresponding to the complete orbit $z(t)$} through $x$ is denoted by $\alpha_{z}(x)$, and defined to be the following:
$$\alpha_{z}(x):= \left\{ y\in X: \exists \ t_n \downarrow -\infty \ \ \text{such that  } z(t_n)\rightarrow y \right\}.$$
A set $M\subset X$ is said to be \emph{forward invariant} if $S(t)M\subset M$ for all $t\geq 0$.  A set $M\subset X$ is said to be \emph{invariant} if $S(t)M=M$ for all $t\geq 0$.  The following equivalent definition will be important:  $M$ is invariant if and only if, for any $x\in M$, a complete orbit through $x$ exists and $\gamma(x)\subset M$.  \\
 The \emph{stable manifold} of a compact invariant set $A$ is denoted by $W^s$ and is defined as
$$W^s(A)= \left\{x\in X: \omega(x)\neq\emptyset \text{ and } \omega(x)\subset A \right\}.$$
The \emph{unstable manifold} is defined by
\begin{align*}
W^u(A)= \left\{x\in X: \text{there exists a backward orbit } z(t) \text{ through } x, \alpha_{z}(x)\neq\emptyset \text{ and } \alpha_z(x)\subset A \right\}.
\end{align*}

Now, we prove two propositions concerning limit sets in forward and backward time, respectively.  First, we prove a simple result about the stable manifold of the singleton $\left\{E_i\right\}$, $W^s(\left\{E_i\right\})$, which will be applied later in the proof of uniform persistence for our system.
\begin{proposition}\label{omega}
Let $x\in X$.  If $x\in W^s(\left\{E_i\right\})$, then $S(t)x\rightarrow E_i$ as $t\rightarrow\infty$. \end{proposition}
\begin{proof}
We will show $x\in W^s(\left\{E_i\right\}) \Rightarrow \lim_{t\rightarrow\infty}S(t)x=E_i$.  Suppose by way of contradiction, $\exists \epsilon>0, t_n\uparrow\infty$ such that $\left\|S(t_n)x-E_i\right\|\geq\epsilon$.  As shown in the proof of Proposition \ref{smooth}, the semigroup $S(t)$ can be written as $S(t)=C_i(t)+U_i(t)$.  Since $C_i(t)\left\{S(t_n)\right\}$ is pre-compact, there exists a convergent subsequence:  $C_i\left(t_{n_k}\right)\rightarrow x^*$.  Then $S\left(t_{n_k}\right)\rightarrow x^*$ because $\left\|U_i\left(t_{n_k}\right)\right\|\rightarrow 0$.  But then $x^*\in\omega(x)$, but $x^*\neq E_1$, which is a contradiction to the definition of the stable manifold. 
\end{proof}

Second, we consider the alpha limit set corresponding to a complete orbit corresponding to solutions of the model (\ref{a3}).  The following result is utilized in the application of a Lyapunov functional to our system.
\begin{proposition}\label{alpha}
Let $x\in X$ and consider the model (\ref{a3}).  If there is complete orbit $z(t)$ through $x$, then the set $\left\{z(t):t\in\mathbb{R}\right\}$ is pre-compact, and $\alpha_z(x)$ is non-empty, compact, and invariant.  In addition, if $\alpha_z(x)=\left\{E_i\right\}$, then $z(t)\rightarrow E_i$ as $t\rightarrow -\infty$.
\end{proposition}
\begin{proof}
Suppose that $z(t)$ is a complete orbit through $x$.  To show that $\left\{z(t):t\in\mathbb{R}\right\}$ is pre-compact, we can modify the arguments in the proof of Propositon \ref{smooth}.  Verifying condition (i) from Theorem \ref{compact} is essentially the same as in Proposition \ref{smooth}.  Next consider condition (ii) from Theorem \ref{compact} applied to the complete orbit: $\lim_{h\rightarrow\infty}\int_h^{\infty}T^*_i(t,a)\,da$ converges to zero uniformly for all $t\in\mathbb{R}$.   
$$\lim_{h\rightarrow\infty}\int_h^{\infty}T^*_i(t,a)=\lim_{h\rightarrow\infty}\int_h^{\infty}k\phi_i(a)V_i(t-a)T(t-a)\,da\leq kM^2 \lim_{h\rightarrow\infty}\int_h^{\infty}e^{-ba}=0$$
Clearly the convergence is uniform $\forall t\in\mathbb{R}$, so $\left\{z(t):t\in\mathbb{R}\right\}$ is pre-compact.  Then, $\alpha_z(x)$ is non-empty and compact.  The remainder of the theorem conclusions follow from Theorem 2.48 in \cite{SmTh}.
\end{proof}

Next, we recall definitions and a result about global attractors.  A set $A\subset X$ \emph{attracts} a set $B\subset X$ if, ${\rm dist}(S(t)B,A)\rightarrow 0$ as $t\rightarrow \infty$, where ${\rm dist}(B,A)$ is the distance from set $B$ to set $A$, i.e.
$${\rm dist}(B,A):=\sup_{y\in B}\inf_{x\in A} \left\|y-x\right\|.$$ 
A set $A$ in $X$ is defined to be an \emph{attractor} if $A$ is non-empty, compact and invariant, and there exists some open neighborhood $U$ of $A$ in $X$ such that $A$ attracts $U$.   A \emph{global attractor} is defined to be an attractor which attracts every point in $X$.  A set $A\subset X$ is said to be a \textit{strong global attractor} if it is a global attractor, and in addition, for any bounded set $B\subset X$, $A$ attracts $B$.  

The following theorem gives a sufficient condition for existence of a strong global attractor.
\begin{theorem}[Hale, \cite{hale1}]\label{attractor}
If $S(t)$ is asymptotically smooth and point dissipative in $X$, and if the forward orbit of bounded sets is bounded in $X$, then there is a strong global attractor $A$ in $X$.
\end{theorem}

Proposition \ref{smooth} and Proposition \ref{bounded} show that the semigroup $S(t)$ generated by the system (\ref{a3}) is asymptotically smooth and point dissipative on the state space $X$.  We also notice that the argument in the proof of Proposition \ref{bounded} implies that the forward orbit of bounded sets is bounded in $X$.  Thus, by Theorem \ref{attractor}, we arrive at the following proposition.

\begin{proposition}
Let $S(t)$ be the semigroup generated by the system (\ref{a3}) on the state space $X$ defined previously.  There is a strong global attractor $A$ in $X$.
\end{proposition}

\section{Reproduction Numbers and Extinction Condition}\label{s4}
\subsection{Reproduction numbers and equilibria}\label{s41}
There exists a unique disease-free equilibrium, $E_0$, for the system (\ref{a3}) with $E_0=(\overline{T}_0,0,....,0)$. \\
For $i=1,,.n$, define
$$N_i=\int_0^{\infty} p_i(a)\phi_i(a)\,da.$$
 $N_i$ is the average number of virions produced by an infected cell that is infected with strain $i$.  Define the basic reproduction number for strain $i$ as
\begin{align*}
\mathcal{R}_i=\frac{N_ik_i\overline{T}_0}{\gamma_i}\tag{5}\label{rn}
\end{align*}
Thus, $\mathcal R_i$ is intuitively the average amount of secondary infected cells induced by a single infected cell for strain $i$ in a population of target cells at carrying capacity $\overline{T}_0$.

Now we determine non-trivial equilbria.  First notice from (\ref{a3}) that a infected cell equilibrium density, $\overline{T}^*_i(a)$, satisfies $\overline{T}^*_i(a)=k_i\overline{V}_i\overline{T}\phi_i(a)$ where $\overline{V}_i$, $\overline{T}$ are equilibrium values of the components $V_i$ and $T$ respectively.  By setting the ODEs in (\ref{a3}) to zero, we obtain $\overline{T}=\frac{\overline{T}_0}{\mathcal{R}_i}$ if $\overline{V}_i\neq 0$.  It is then readily observed that for each strain,  there exists the single strain equilibrium $E_i=(\overline{T}_i,0,..,0,\overline{V}_i,0,...,0,\overline{T}^*_i(a),0,..,0)$, where
$$\overline{T}_i=\frac{\overline{T}_0}{\mathcal{R}_i}, \quad \overline{V}_i=\frac{f(\overline{T}_i)}{k\overline{T}_i}, \quad \overline{T}^*_i(a)=k_i\overline{V}_i\overline{T}_i\phi_i(a).$$
Here $E_i$ is biologically relevant, i.e. $\overline{V}_i>0, \int_0^{\infty} \overline{T}^*_i(a) \,da>0$, if and only if $\mathcal{R}_i>1$. 

If $\mathcal R_i\neq \mathcal R_j$ for all $i\neq j$, then there are no coexistence equilibria.  However, when \\ $\mathcal R_{i_1}=\mathcal R_{i_2}=\cdots=\mathcal R_{i_{\ell}}>1$, there exists a $\ell-1$ dimensional hyperplane of coexistence equilibria described by the equation: $$0=f\left(\overline{T}_{i_1}\right)-\overline{T}_{i_1}\sum_{j=1}^{\ell} k_{i_j}\overline{V}_{i_j}.$$  

\subsection{Extinction of strain $i$ when $\mathcal{R}_i<1$}
The following theorem establishes extinction of virus strain $i$ if $\mathcal{R}_i<1$.
\begin{proposition}
If $\mathcal{R}_i<1$, then $V_i(t)\rightarrow 0$ and $\int_0^{\infty}T^*_i(t,a)\,da\rightarrow 0$ as $t\rightarrow\infty$.
\end{proposition}
\begin{proof}
Let $T^{\infty}:=\limsup_{t\rightarrow \infty}T(t)$ and $V_i^{\infty}:=\limsup_{t\rightarrow \infty}V_i(t)$.  By assumption, $f(T)<0$ if $T>\overline{T}_0$; thus,  $\frac{dT}{dt}\leq f(T)<0$ when $T>\overline{T}_0$.  Then the previous statement, along with the smoothness of $f(T)$, imply that $T^{\infty}\leq \overline{T}_0$.  For all $\epsilon>0$, there exists $\tau>0$ such that $\forall t\geq \tau, \ \ \ T(t)\leq \overline{T}_0+\epsilon, \ V_i(t)\leq V_i^{\infty}+\epsilon$.  Also note that
\begin{align*}
 \int_t^{\infty} \! p_i(a) \frac{\phi_i(a)}{\phi_i(a-t)}T^*_i(0,a-t) \,da &\leq \kappa e^{-bt}\int_0^{\infty}\!T^*_i(0,a) \\*
& \ \ \rightarrow 0 \ \text{ as } t\rightarrow \infty
\end{align*}
Hence, we can pick the $\tau>0$, such that $\int_t^{\infty} \! p_i(a)  \frac{\phi_i(a)}{\phi_i(a-t)}T^*_i(0,a-t) \,da < \epsilon$.
By the semigroup property, we can without loss of generality assume $\tau=0$.  Then
\begin{align*}
\dot{V_i} &= \int_0^t\! k_iV_i(t-a)T(t-a)p_i(a)\phi_i(a)\,da + \int_t^{\infty} \! p_i(a) \frac{\phi_i(a)}{\phi_i(a-t)}T^*_i(0,a-t) \,da -\gamma_i V_i \\
&\leq k_i(\overline{T}_0+\epsilon)(V_i^{\infty}+\epsilon) \int_0^t\! p_i(a)\phi_i(a)\,da + \epsilon -\gamma_i V_i \\
& =  \frac{(\overline{T}_0+\epsilon)(V_i^{\infty}+\epsilon)\gamma_i \mathcal{R}_i}{\overline{T}_0} + \epsilon -\gamma_i V_i.
\end{align*}
Hence,
\begin{align*}
V_i^{\infty}&\leq \frac{(\overline{T}_0+\epsilon)(V_i^{\infty}+\epsilon) \mathcal{R}_i}{\overline{T}_0} +   \frac{\epsilon}{\gamma_i}.
\end{align*}
Because $\mathcal{R}_i<1$, from the above inequality, we find that for $\epsilon>0$ sufficiently small, if $V_i^{\infty}>0$, then $V_i^{\infty}<V_i^{\infty}$.  This is a contradiction which forces $V_i^{\infty}=0$.  Then with a similar reset of time argument using the semigroup property, we find that
\begin{align*}
\limsup_{t\rightarrow \infty}\int_0^{\infty} \! T^*_i(t,a) \,da \leq \limsup_{t\rightarrow\infty} \left(\int_0^t\! k_i\epsilon (\overline{T}_0+\epsilon) \phi_i(a)\,da + \int_t^{\infty} \!   \frac{\phi_i(a)}{\phi_i(a-t)}T^*_i(0,a-t) \,da\right) \\*
\leq \epsilon k_i(\overline{T}_0+\epsilon) \int_0^{\infty}\!e^{-ba}\,da + \epsilon \\*
\ \rightarrow 0 \text{  as  } \epsilon\rightarrow 0
\end{align*}
Hence, $V_i(t)\rightarrow 0$ in $\mathbb{R}_+$ and $T^*_i(t,a)\rightarrow 0$ in $L^1_+(0,\infty)$ as $t\rightarrow\infty$.  
\end{proof}

\section{Competitive Exclusion}\label{s5}
In this section, we will prove that $E_1$, the equilibrium corresponding to the strain with the maximum reproduction number, is globally attracting, i.e. the competitive exclusion principle holds.   In order to prove the result we will use strong mathematical induction in order to establish uniform persistence and apply a Lyapunov functional argument, but we need to establish several results first.  

We will consider the case where the viral strains all have different reproduction numbers which are greater than 1.   Note that all of the following results hold for the case where some viral strains have reproduction number less than unity, but in order to make the notation simpler, we assume that $\min_i \mathcal{R}_i>1$.

Without loss of generality, suppose that
\begin{align*}
\mathcal{R}_1>\mathcal{R}_2>.....>\mathcal{R}_n>1\label{repro}\tag{8}
\end{align*}
Another way of writing the above condition is the following:
\begin{align*}
0<\overline{T}_1<\overline{T}_2<.....<\overline{T}_n<\overline{T}_0 \label{repro1}\tag{9}
\end{align*}
Then, as shown in Section \ref{s41} there are the $n$ single strain equilibria: $E_i$,  $i=1,...,n$ and no coexistence equilibria.  In the case where some reproduction numbers are equal, the rigorous analysis is more difficult, but we can conjecture the dynamics. If $\mathcal R_1$, the largest reproduction number, is distinct but $\mathcal R_i=\mathcal R_{i+1}$ for some $ i>1$, we expect competitive exclusion, i.e. global convergence to $E_1$, as in the ODE case \cite{deleenheerandpilyugin}.  However, the subsequent induction argument utilized for the global analysis of the model does not apply to this case.  When $\mathcal R_1=\mathcal R_2=\cdots=\mathcal R_{\ell}$, the situation is more complex.  In this case, we conjecture that the global attractor is the $\ell -1 $ dimensional hyperplane of coexistence equilibria. 

Suppose also that $f(T)$ satisfies the sector condition for all $\overline{T}_i, \ \ i=1,...,n$:
\begin{align*}
    \left(f(T)-f(\overline{T}_i)\right) \left(1-\frac{\overline{T}_i}{T}\right) \leq 0. \tag{10}\label{(C)}
\end{align*}
The sector condition was introduced by De Leenheer and Pilyugin, in order to prove global stability of the infection equilibrium in the single-strain and multi-strain ODE standard virus model \cite{deleenheerandpilyugin}.  Note that this condition is satisfied when $f(T)$ is a decreasing function, independently of the value of $\overline{T}_i$, for example $f(T)=f_1(T)=s-c T$.  In the case of $f(T)=f_2(T)=s-cT+rT(1-T/T_{max})$, Condition (\ref{(C)}) is satisfied when $s\geq f(\overline{T}_i)$.

\subsection{Lyapunov functional}
In order to analyze the global dynamics via a Lyapunov functional, we consider complete orbits for our system.
Let $x\in X$.  Suppose that we can find a complete orbit $z(t)$ through $x$.  Suppose that $z(t)=\left((T(t),V_1(t),.....,V_n(t),T^*_1(t,a),.....,T^*_n(t,a)\right)$, where $t\in\mathbb{R}$.  Then $z(t)$ must satisfy the following system for all $t\in\mathbb{R}$:
\begin{align*}
\frac{dT(t)}{dt} &= f(T(t))- \sum_{i=1}^nk_iV_i(t)T(t), \\
 \frac{dV_i(t)}{dt} &= \int_0^{\infty} \!p_i(a)T_i^*(t,a)\,da- \gamma V_i(t), \\
T_i^*(t,a) &= \phi_i(a) k_iV_i(t-a)T(t-a) \\
T_i^*(t,0)&=k_iV_i(t)T(t)
 \end{align*}

In the proof of the following proposition, we find a Lyapunov functional for a complete orbit $z(t)$, which is well-defined and bounded when $z(t)$ satisfies certain criteria, namely $z(t)$ is bounded from above and bounded away from an appropriate boundary set.  Under these criteria, a LaSalle invariance type argument can be invoked to show that the complete orbit $z(t)$ must be in fact be the equilibrium $E_1$.

\begin{proposition}\label{lyap}
Let $0<\epsilon<M<\infty$ be arbitrary.  Suppose that $x\in X$ and there exists a complete orbit $z(t)$ through $x$ such that $\left\| z(t)\right\|\leq M , \ V_1(t)\geq \epsilon, \ T(t)\geq\epsilon \ \forall t\in\mathbb{R}$.  Then, $x=E_1$.
\end{proposition}
 
\begin{proof}
We introduce a transformation which will make certain calculations simpler.  For $x\in X$, define the transformation, $h(x)$ as:
$$h(x)=\left(T,V_1,....,V_n,\frac{1}{\phi_1(a)}T^*_1(a),....,\frac{1}{\phi_n(a)}T^*_n(a)\right) $$
Let $x\in X$ with complete orbit $z(t)$ through $x$.  Then $h(z(t))$ satisfies:
\begin{align*}
\frac{dT(t)}{dt} &= f(T(t))- \sum_{i=1}^nk_iV_i(t)T(t), \\
 \frac{dV_i(t)}{dt} &= \int_0^{\infty} \!q_i(a)u_i(t,a)\,da- \gamma V_i(t), \tag{11} \\
u_i(t,a) &=  k_iV_i(t-a)T(t-a) \\
u_i(t,0)&=k_iV_i(t)T(t)
 \end{align*}
where 
\begin{align*}
u_i(t,a)&=\frac{1}{\phi_i(a)}T_i^*(t,a) \ \text{ and} \\ 
q_i(a)&=\phi_i(a)p_i(a) \tag{12}
\end{align*}
 Also define the transformed components of the equilibria, $E_i$ by 
 $$\overline{u}_i= \frac{1}{\phi_i(a)}\overline{T}^*_i(a)=k\overline{V}_i\overline{T}_i .$$
  Notice that $\overline{u}_i^i$ is a constant function, i.e. does not vary with $a$.   
Define the following function on $(0,\infty)$:
\begin{align*}
g(x)=x-1-\log(x) \tag{13}
\end{align*}
  Note that $g(x)$ is non-negative and continuous on $(0,\infty)$ with a unique root at $x=1$.  Let 
  \begin{align*}
  \alpha_i(a)=\int_a^{\infty}\!q_i(\ell) \,d\ell .\tag{13}
  \end{align*}
By the Lebesgue Differentiation Theorem, $\alpha_i(a)$ is differentiable with
  \begin{align*}
  \alpha_i'(a)=-q_i(a).\tag{14}
  \end{align*}
Define the following ``candidate'' Lyapunov functional expression on $h(X)$:
$$W:(T,V_1,...,V_n,u_1(a),...,u_n(a))\mapsto W_T+W_{V_1}+W_{u_1}+W_{\partial}$$ where
\begin{align*}
&W_T=\frac{\overline{T}_1}{\overline{u}_1}g\left(\frac{T}{\overline{T}_1}\right),\quad W_{V_1}=\frac{k_1\overline{T}_1\overline{V}_1}{\gamma_1 \overline{u}_1}g\left(\frac{V_1}{\overline{V}_1}\right), \tag{15} \\
 & W_{u_1}=\frac{k_1\overline{T}_1}{\gamma_1}\int_0^{\infty}\!\alpha_1(a)g\left(\frac{u_1(a)}{\overline{u}_1}\right)\,da, \quad W_{\partial}=\frac{1}{\overline{u}_1} \sum_{i=2}^n{ \frac{1}{\alpha_i(0)}\left(\int_0^{\infty}\!\alpha_i(a)u_i(a)\,da+V_i\right)}
\end{align*}
Note that the composition $W\circ h$ is certainly not well-defined on all of $X$.   However, we simply want it to be well defined and bounded for a complete orbit $z(t)$ that is bounded from above and away from the appropriate boundary set.

Suppose that $x\in X$ and there exists a complete orbit $z(t)$ through $x$ such that $\left\| z(t)\right\|\leq M , \ V_1(t)\geq \epsilon, \ T(t)\geq\epsilon \ \forall t\in\mathbb{R}$.  Then, $k_1\epsilon^2 \leq u_1(t,a) \leq k_1M^2$ for all $a\in[0,\infty)$ and $t\in\mathbb{R}$.  Hence, $\exists M_1>0$ such that $$0\leq g\left(\frac{u_1(t,a)}{\overline{u}_1}\right)\leq M_1$$
Then,
\begin{align*}
\int_0^{\infty}\!\alpha_1(a)g\left(\frac{u_1(t,a)}{\overline{u}_1}\right)\,da &\leq M_1 \int_0^{\infty}\int_a^{\infty}\! \phi_1(\ell)p_1(\ell)\,d\ell \,da \\
&\leq \kappa M_1 \int_0^{\infty}\int_a^{\infty}\! e^{-b\ell}\,d\ell \,da \\
&= \frac{\kappa M_1}{b^2} <\infty
\end{align*}
Also,
\begin{align*}
\int_0^{\infty}\alpha_i(a)u_i(t,a)\,da\leq M^2\int_0^{\infty}\alpha_i(a)\,da\leq \frac{\kappa M^2}{b^2} 
\end{align*}
Therefore it follows that $W=W_T+W_{u_1}+W_{V_1}+W_{\partial}$ is well-defined and bounded on the transformed complete orbit $z(t)$.  For convenience, $W_T(T(t))$ is denoted by $W_T$, and likewise for the other components.  We also note that $\int_0^{\infty}\!\alpha_1(a)g\left(u_1(t,a)\right)\,da$ and $\int_0^{\infty}\!\alpha_i(a)u_i(t,a)\,da$ are differentiable in $t$ since they are convolutions which we can differentiate, as we will see below.  Hence, $W(h(z(t)))$ is differentiable in $t$.
\begin{align*}
\frac{d}{dt}W_{u_1}&=\frac{d}{dt} \ \frac{k_1\overline{T}_1}{\gamma_1}\int_0^{\infty}\!\alpha_1(a)g\left(\frac{u_1(t,a)}{\overline{u}_1}\right) \,da \\
&=  \frac{k_1\overline{T}_1}{\gamma_1} \ \frac{d}{dt}\int_0^{\infty}\!\alpha_1(a)g\left(\frac{u_1(t-a,0)}{\overline{u}_1}\right) \,da \\
&=  \frac{k_1\overline{T}_1}{\gamma_1} \ \frac{d}{dt}\int_{-\infty}^t\!\alpha_1(t-s)g\left(\frac{u_1(s,0)}{\overline{u}_1}\right) \,ds \\
&=  \frac{k_1\overline{T}_1}{\gamma_1} \left[ \alpha_1(0)g\left(\frac{u_1(t,0)}{\overline{u}_1}\right)+\int_{-\infty}^t\!\alpha_1\p(t-s)g\left(\frac{u_1(s,0)}{\overline{u}_1}\right) \,ds \right]\\
&=  \frac{k_1\overline{T}_1}{\gamma_1} \left[ \alpha_1(0)g\left(\frac{u_1(t,0)}{\overline{u}_1}\right)+\int_0^{\infty}\!\alpha_1\p(a)g\left(\frac{u_1(t,a)}{\overline{u}_1}\right) \,da \right]\\
&=  \frac{k_1\overline{T}_1}{\gamma_1} \left[ \int_0^{\infty}\!q_1(a) \left(\frac{u_1(t,0)}{\overline{u}_1}-1-\log \frac{u_1(t,0)}{\overline{u}_1}-\frac{u_1(t,a)}{\overline{u}_1}+1+\log \frac{u_1(t,a)}{\overline{u}_1} \right) \,da \right] \\
&=  \frac{k_1\overline{T}_1}{\gamma_1} \left[ \int_0^{\infty}\!q_1(a) \left(\frac{u_1(t,0)}{\overline{u}_1}-\frac{u_1(t,a)}{\overline{u}_1}+\log \frac{u_1(t,a)}{u_1(t,0)} \right) \,da \right] \\
\end{align*}
We use the following equilibrium conditions in the next calculation:
$$ f(\overline{T}_1)=k_1\overline{T}_1\overline{V}_1=\overline{u}_1, \quad \frac{\overline{V}_1}{V_1}=\frac{T\overline{u}_1}{\overline{T}_1u_1(t,0)}, \quad \frac{\gamma_1}{k_1\overline{T}_1}=\int_0^{\infty}\! q_1(a)\,da .$$
\begin{align*}
\frac{d}{dt}&\left(W_T+W_{V_1}+W_{\partial}\right) \\*
&=\frac{d}{dt}\left[\frac{\overline{T}_1}{\overline{u}_1}g\left(\frac{T}{\overline{T}_1}\right)+\frac{k_1\overline{T}_1\overline{V}_1}{\overline{u}_1\gamma_1}g\left(\frac{V_1}{\overline{V}_1}\right)+\frac{1}{\bar{u_1}^1} \sum_{i=2}^n{ \frac{1}{\alpha_i(0)}\left(\int_0^{\infty}\!\alpha_i(a)u_i(t,a)\,da+V_i\right)}\right] \\
&=\frac{1}{\overline{u}_1}\left[\overline{T}_1\cdot g\p\left(\frac{T}{\overline{T}_1}\right)\cdot \frac{\dot{T}}{\overline{T}_1} +\frac{k_1\overline{T}_1\overline{V}_1}{\gamma_1}g\p\left(\frac{V_1}{\overline{V}_1}\right) \frac{\dot{V_1}}{\overline{V}_1}+  \sum_{i=2}^n{ \frac{1}{\alpha_i(0)}\left(\frac{d}{dt}\int_{-\infty}^t\!\alpha_i(t-s)u_i(s,0)\,ds+\dot{V}_i\right)}\right]\\
&=  \frac{1}{\overline{u}_1}\left[ \left(1-\frac{\overline{T}_1}{T}\right)  \left(f(T)-k_1V_1T-\sum_{i=2}^n{k_iV_iT}\right) + \frac{k_1\overline{T}_1}{\gamma_1}\left(1-\frac{\overline{V}_1}{V_1}\right) \left(\int_0^{\infty}\!q_1(a)u_1(t,a)\,da -\gamma_1 V_1\right) \right.  \\
& \qquad +\left.  \sum_{i=2}^n{ \frac{1}{\alpha_i(0)}\left(\alpha_i(0)u_i(t,0)+\int_{-\infty}^t\!\alpha_i'(t-s)u_i(s,0)\,ds+\int_0^{\infty}\!\alpha_i(a)u_i(t,a)\,da-\gamma_iV_i\right)}\right] \\
&=  \frac{1}{\overline{u}_1}\left[ \left(f(T)-f(\overline{T}_1)\right) \left(1-\frac{\overline{T}_1}{T}\right)  +f(\overline{T}_1)\left(1-\frac{\overline{T}_1}{T}\right)+ \left(1-\frac{\overline{T}_1}{T}\right)  \left(-\sum_{i=2}^n{k_iV_iT}\right)-k_1V_1T +k_1V_1\overline{T}_1 \right. \\
& \qquad \left. + \frac{k_1\overline{T}_1}{\gamma_1} \int_0^{\infty}\!q_1(a)u_1(t,a)\left(1-\frac{\overline{V}_1}{V_1}\right)\,da -k_1V_1\overline{T}_1 + \frac{k_1\overline{T}_1}{\gamma_1} \gamma_1\overline{V}_1 \right.  \\
& \qquad  + \left. \sum_{i=2}^n{ \frac{1}{\alpha_i(0)}\left(\alpha_i(0)k_iV_iT-\int_0^{\infty}\!\alpha_i(a)u_i(t,a)\,da+\int_0^{\infty}\!\alpha_i(a)u_i(t,a)\,da-\alpha_i(0)k_i\overline{T}_iV_i\right)}\right] \\
&=\frac{1}{\overline{u}_1} \left(f(T)-f(\overline{T}_1)\right) \left(1-\frac{\overline{T}_1}{T}\right)  + \frac{1}{\overline{u}_1} \frac{k_1\overline{T}_1}{\gamma_1} \left[ \frac{\gamma_1}{k_1\overline{T}_1}\left(f(\overline{T}_1)-f(\overline{T}_1)\frac{\overline{T}_1}{T}-k_1V_1T\right) \right] \\
&\qquad + \frac{1}{\overline{u}_1} \sum_{i=2}^n{ \left[\left(\frac{\overline{T}_1}{T}-1\right)  k_iV_iT +k_iV_iT-k_i\overline{T}_iV_i\right]} + \frac{1}{\overline{u}_1} \frac{k_1\overline{T}_1}{\gamma_1}\left[\int_0^{\infty}\!q_1(a)u_1(t,a)\left(1-\frac{\overline{V}_1}{V_1}\right)\,da + \frac{\gamma_1}{k_1\overline{T}_1}k_1\overline{T}_1\overline{V}_1 \right]\\
&= \frac{1}{\overline{u}_1} \left(f(T)-f(\overline{T}_1)\right) \left(1-\frac{\overline{T}_1}{T}\right) +\frac{1}{\overline{u}_1}\sum_{i=2}^n { k_iV_i\left(\overline{T}_1-\overline{T}_i\right)} \\
& \qquad +  \frac{k_1\overline{T}_1}{\gamma_1}  \int_0^{\infty}\!q_1(a)\left( \frac{-u_1(t,0)}{\overline{u}_1}-\frac{\overline{T}_1}{T}+\frac{u_1(t,a)}{\overline{u}_1}-\frac{T u_1(t,a)}{\overline{T}_1 u_1(t,0)} +2 \right)\,da
\end{align*}
Therefore,
\begin{align*}
\frac{d}{dt}\left(W_T+W_{V_1}+W_{u_1}+W_{\partial}\right)&= \frac{1}{\overline{u}_1} \left(f(T)-f(\overline{T}_1)\right) \left(1-\frac{\overline{T}_1}{T}\right) +\frac{1}{\overline{u}_1}\sum_{i=2}^n { k_iV_i\left(\overline{T}_1-\overline{T}_i\right)} \\
& \qquad +  \frac{k_1\overline{T}_1}{\gamma_1}  \int_0^{\infty}\!q_1(a)\left( 2-\frac{\overline{T}_1}{T}-\frac{T u_1(t,a)}{\overline{T}_1 u_1(t,0)} + \log \frac{u_1(t,a)}{u_1(t,0)} \right)\,da  \\
&\leq - \frac{k_1\overline{T}_1}{\gamma_1}  \int_0^{\infty}\!q_1(a)\left( g\left(\frac{\overline{T}_1}{T}\right) + g\left(\frac{T u_1(t,a)}{\overline{T}_1 u_1(t,0)} \right)\right)\,da  \\
&\leq 0.
\end{align*}
Here we have used the sector condition (Condition (\ref{(C)})), the fact that $\overline{T}_1<\overline{T}_i \ \forall i\geq 2$, and the positivity of $g$.  Hence, we find that
\begin{align*}
\frac{dW}{dt}=0 &\Leftrightarrow u_1(t,a)=u_1(t,0) \ \ \text{and} \ \ \frac{\overline{T}_1}{T}=\frac{Tu_1(t,a)}{\overline{T}_1u_1(t,0)}  \ \ \text{and  } V_i=0 \ \forall i\geq 2 \\
&\Leftrightarrow u_1(t,a)=u_1(t,0) \ \ \text{and} \ \ T=\overline{T}_1  \ \ \text{and  } V_i=0 \ \forall i\geq 2\\
&\Leftrightarrow \frac{d}{dt}u_1(t,0)=0 \ \ \text{and} \ \ \frac{d}{dt}T=0  \ \ \text{and  } V_i=0 \ \forall i\geq 2 \\
&\Leftrightarrow  \frac{d}{dt}k_1V_1T=0 \ \ \text{and} \ \ \frac{d}{dt}T=0  \ \ \text{and  } V_i=0 \ \forall i\geq 2\\
&\Leftrightarrow \frac{d}{dt}V_1=0 \ \ \text{and} \ \ \frac{d}{dt}T=0  \ \ \text{and  } V_i=0 \ \forall i\geq 2.
\end{align*}
Hence, the maximal invariant set with the property that $\frac{dW}{dt}=0$ on this set is $\left\{ E_1\right\}$. Note that the same result holds in the case $n=1$, with $W^1:=W_T+W_{V_1}+W_{u_1}$.  \\
By Proposition \ref{alpha}, $\alpha_z(x)$ is compact, non-empty, and invariant.  Let $\widetilde{x} \in \alpha_z(x)$.   Let $z(t)=(T(t),V(t),T^*_1(t,a))$.  Then $\exists t_n \downarrow -\infty$ such that $x_n:=z(t_n)\rightarrow \widetilde{x}$.  In particular $T^*_1(t_n,a)\rightarrow \widetilde{T}^*_1(a)$ in $L^1$ as $t_n\downarrow -\infty$.  Then, we claim $W_{u_1}\left(\frac{1}{\phi_1(a)}T^*_1(t_n,a)\right)\rightarrow W_{u_1}(\frac{1}{\phi_1(a)}\widetilde{T}^*_1(a))$ in $L^1$ as $t_n\downarrow -\infty$.   Indeed,
\begin{align*}
\left|W_{u_1}(u_1(t_n,a)-W_{u_1}(\widetilde{u}_1(t,a)\right|&= \left|\int_0^{\infty}\int_a^{\infty}\!\phi_1(\ell)p_1(\ell)\,d\ell \left( g\left(\frac{u_1(t_n,a)}{\overline{u}_1}\right)-  g\left(\frac{\widetilde{u}_1(a)}{\overline{u}_1}\right) \right) \,da \right| \\
&\leq \kappa \int_0^{\infty}\int_a^{\infty}\!\phi_1(\ell)\,d\ell \left| g\left(\frac{u_1(t_n,a)}{\overline{u}_1}\right)-  g\left(\frac{\widetilde{u}_1(a)}{\overline{u}_1}\right) \right| \,da  \\
&\leq  \kappa \int_0^{\infty}\int_a^{\infty}\!\phi_1(\ell)\,d\ell \max_{k\epsilon^2\leq s\leq kM^2}|g\p_1(s)| \cdot \left| \frac{u_1(t_n,a)}{\overline{u}_1}-  \frac{\widetilde{u}_1(a)}{\overline{u}_1} \right| \,da  \\
&=  \frac{\kappa M_1}{\overline{u}_1} \int_0^{\infty}\int_a^{\infty}\!\phi_1(\ell)\,d\ell \frac{1}{\phi_1(a)}\left| T^*_1(t_n,a)-\widetilde{T}^*_1(a) \right| \,da  \\
&=  \frac{\kappa M_1}{\overline{u}_1} \int_0^{\infty}\int_a^{\infty}\!e^{-\int_a^{\ell}\!\delta_1(s)\,ds}\,d\ell  \left| T^*_1(t_n,a)-\widetilde{T}^*_1(a) \right| \,da  \\
&\leq  \frac{\kappa M_1}{\overline{u}_1} \int_0^{\infty}\int_a^{\infty}\!e^{-b(\ell-a)}\,d\ell  \left| T^*_1(t_n,a)-\widetilde{T}^*_1(a) \right| \,da  \\
&= \frac{\kappa M_1}{b\overline{u}_1} \int_0^{\infty}\!  \left| T^*_1(t_n,a)-\widetilde{T}^*_1(a) \right| \,da  \\
&\rightarrow 0 \ \ \ \text{as   } t_n\downarrow -\infty .
\end{align*}
In a similar way, along with using the continuity of $V_1$, we can obtain $W_{\partial}$ is continuous for $h(z(t))$.  The convergence of the other components of $W\circ h$ is a consequence of the continuity of $g$.  Then, $W\circ h(z(t_n))\rightarrow W\circ h(\widetilde{x})$ as $t_n\downarrow-\infty$.  Since $W\circ h(z(t))$ is a non-increasing map, which is bounded above, we conclude that $W\circ h(z(t))\uparrow c <\infty$ as $t\downarrow -\infty$.  Therefore, $W\circ h(\hat{x})=c$ for all $\hat{x}\in\alpha_z(x)$.  Combining this with the fact that $\alpha_z(x)$ is invariant, we get that $W\circ h(\zeta(t))=c$ for all $t\in\mathbb{R}$, where $\zeta(t)$ is a complete orbit through $\widetilde{x}$ (with $\zeta(0)=\widetilde{x}$).  Hence, $\frac{d}{dt}W\circ h(\zeta(t))=0$ for all $t\in\mathbb{R}$.  This implies that $h\left(\left\{z(t): t\in\mathbb{R} \right\}\right)$ is an invariant set with the property that $\frac{dW}{dt}=0$.  Therefore, $h(\zeta(t))=h(E_1)$ for all $t$, in particular when $t=0$.  So, $\widetilde{x}=E_1$.  This shows that $\alpha_z(x)=\left\{E_1\right\}$.   Thus, $W\circ h(z(t))\leq W\circ h(E_1)$ for all $t\in\mathbb{R}$.  Since $E_1$ is the unique minimizer of $W\circ h$, $z(t)=E_1 \ \forall t\in \mathbb{R}_+$, and hence $x=E_1$.
\end{proof}
Thus, if $z(t)$ is a complete orbit such that $\left\| z(t)\right\|\leq M, \ V_1(t)\geq \epsilon, \ T(t)\geq\epsilon \ \forall t\in\mathbb{R}$ for some $\epsilon>0$, then $z(t)=E_1 \ \forall t\in\mathbb{R}$.  

\subsection{Persistence theory}
Proposition \ref{lyap} states that the only complete orbit in an appropriate subset ($z(t)$ satisfies hypothesis in Proposition \ref{lyap}) for the system (\ref{a3}) is the equilibrium $E_1$.   If we can find a global attractor on this appropriate subset, then due to its invariance, the global attractor will reduce to the equilibrium $E_1$.  To follow this strategy, we utilize persistence theory, in particular a result from Hale and Waltman on uniform persistence \cite{persist} and a result from Magal and Zhao on existence of an interior global attractor \cite{magal2}.

Persistence theory provides a mathematical formalism for determining whether a species will ultimately go extinct or persist in a dynamical model.   Consider $X$ as the closure of an open set $X_1$; that is, $X=X_1\cup \partial X_1$, where $\partial X_1$ (assumed to be non-empty) is the boundary of $X_1$.  Also, suppose that the semigroup $S(t)$ on $X$ satisfies
\begin{align*}
S(t):X_1\rightarrow X_1, \qquad S(t):\partial X_1\rightarrow \partial X_1. \tag{B1}\label{decomp}
\end{align*}
Suppose that $S(t)$ satisfies the conditions of Theorem \ref{attractor}.  Then $S_{\partial}:=S(t)|_{\partial X_1}$ will satisfy the same conditions in $\partial X_1$.  Therefore, there will be a global attractor $A_{\partial}$ in $\partial X_1$. \\
The semigroup $S(t)$ is said to be \emph{uniformly persistent} (with respect to $X_1$ and $\partial X_1$) if there is an $\eta>0$ such that, for any $x\in X_1$,
$$\liminf_{t\rightarrow\infty}d(S(t)x,\partial X_1)\geq \eta .$$

Now we state definitions which will be important in finding a useful equivalent condition to uniform persistence.  A nonempty invariant subset $M$ of $X$ is called an \emph{isolated invariant set} if it is the maximal invariant set of a neighborhood of itself.  The neighborhood is called an \emph{isolating neighborhood}.   Let $M,N$ be isolated invariant sets (not necessarily distinct).  $M$ is said to be \emph{chained} to $N$, written $M\hookrightarrow N$, if there exists an element $x$, $x\notin M\cup N$, such that $x\in W^u(M)\cap W^s(N)$.  A finite sequence $M_1,M_2,....,M_k$ of isolated invariant sets is called a \emph{chain} if $M_1\hookrightarrow M_2\hookrightarrow ....\hookrightarrow M_k \ (M_1\hookrightarrow M_1 \text{ if } k=1)$.  The chain will be called a \emph{cycle} if $M_k=M_1$.

The particular invariant sets of interest are
$$\widetilde{A_{\partial}}=\bigcup_{x\in A_{\partial}}\omega(x).$$
 $\widetilde{A_{\partial}}$ is \emph{isolated} if there exists a covering $M=\cup_{i=1}^k M_k$ of $\widetilde{A_{\partial}}$ by pairwise disjoint, compact, isolated invariant sets $M_1,M_2,...,M_k$ for $S_{\partial}$ such that $M_i$ is also an isolated invariant set for $S(t)$.  $M$ is called an \emph{isolated covering}.  $\widetilde{A_{\partial}}$ will be called \emph{acyclic} if there exists some isolated covering $M=\cup_{i=1}^k M_i$ of $\widetilde{A_{\partial}}$ such that no subset of the $M_i$'s forms a cycle.  An isolated covering satisfying this condition will be called acyclic.

The following theorem will provide the means to prove uniform persistence of the semigroup.
\begin{theorem}[Hale and Waltman, \cite{persist}]\label{semipersist}
Suppose $S(t)$ satisfies Condition (\ref{decomp}) and we have the following:
\begin{itemize}
\item[(i)] $S(t)$ is asymptotically smooth,
\item[(ii)] $S(t)$ is point dissipative in $X$,
\item[(iii)] $\gamma^+(U)$ is bounded if $U$ in $X$,
\item[(iv)] $\widetilde{A_{\partial}}$ is isolated and has an acyclic covering.
\end{itemize}
Then $S(t)$ is uniformly persistent if and only if for each $M_i\in M$
$$W^s(M_i)\cap X_1=\emptyset.$$
\end{theorem}

The following theorem relates uniform persistence to existence of a global attractor in $X_1$.
\begin{theorem}[Magal and Zhao, \cite{magal2}]\label{persistattract}
Assume that the semigroup $S(t)$ satisfies Condition (\ref{decomp}), is asymptotically smooth and uniformly persistent, and has a global attractor $A$.  Then the restriction of $S(t)$ to $X_1$, $S(t)|_{X_1}$, has a global attractor $A_0$.
\end{theorem}

\subsection{Global stability}
In order to proceed, we need to be precise about considering various forward invariant subsets of $X$.  Then, we can define our uniformly persistent set and complementary boundary, and utilize mathematical induction to characterize the dynamics on the boundary set.

First, we define the maximal age of viral production for each strain, which is allowed to be infinity.  Let $$\overline{a}_i=\sup\left\{a\in (0,\infty):p_i(a)>0\right\} \text{  for  } i=1,...,n.$$ We note that $\overline{a}_i$ is allowed to be infinity.  Define the following sets:
  \begin{align*}
   \partial M^0_j&=\left\{\eta(a)\in L^1_+: \int_0^{\overline{a}_j}\!\eta(a)\,da=0 \right\}, \qquad M^0_j=L^1_+\setminus \partial M^0_j \qquad j=1,...,n\\
  \partial X_j&=  \mathbb{R}_+ \times \prod_1^j \left\{0\right\} \times \mathbb{R}_+^{n-i}\times\prod_{i=1}^j\partial M^0_i\times \prod_{j+1}^n L^1_+ \qquad j=1,...,n  \\
   \partial X&=\partial X_n, \qquad X^0=X\setminus \partial X \\ 
   X_1&=X\setminus\partial X_1, \qquad X_j=\left(X\setminus \partial X_j\right) \cap \partial X_{j-1} \quad j=2,...,n \\  
Z_j&= \mathbb{R}_+ \times \mathbb{R}^{j-1}_+\times (0,\infty) \times \mathbb{R}^{n-(j+1)}_+\times\prod_{i=1}^{j-1}L^1_+\times  M^0_j\times \prod_{j+1}^n L^1_+  \quad j=1,...,n\\
   \left(X_j\right)_+&=X_j\cap Z_j \qquad j=1,...,n
  \end{align*}

\begin{proposition}\label{invariance}
For $1\leq j\leq n$, $X_j$ and $\partial X_j$ are forward invariant under the semigroup $S(t)$.  Also, $\partial X$ and $X^0$ are forward invariant, and if $x\in\partial X$, then $S(t)x\rightarrow E_0$ as $t\rightarrow\infty$.   In addition, $S(t)X_j\subset (X_j)_+ \ \forall t>0$.  
\end{proposition}
\begin{proof}
First, we show the conclusions for $\partial X_j, \ j=1,...,n$.  Suppose by way of contradiction that there exists $x\in\partial X_j$ and $t_1>0$ such that $S(t_1)x\in X\setminus\partial X_j$.  Let $\tau=\inf\left\{t>0: S(t)\in X\right\}$.  Since $X\setminus\partial X_j$ is an open set in $X$ and by the continuity of the semigroup $S(t)$, we obtain that $S(\tau)x\in\partial X_j$ and for some $i\leq j$, $V_i(\tau+\epsilon)>0$ or $T^*_i(\tau+\epsilon,a)\in M_j^0$ for all $\epsilon>0$ sufficiently small.  Then for this $i$, the following is true:  
\begin{align*}
\dot{V}_i(\tau)&=\int_0^{\overline{a}_i} \! p_i(a) T^*_i(\tau,a)\,da -\gamma V_i(\tau)=0 \ \text{ and} \\ 
T^*_i(\tau,a)&= \phi_i(a)k_iV_i(\tau-a)T(\tau-a) \mathds{1}_{\left\{\tau>a\right\}} + \frac{\phi_i(a)}{\phi_i(a-t)}T^*_i(0,a-t)\mathds{1}_{\left\{a>\tau\right\}} \\
&=\frac{\phi_i(a)}{\phi_i(a-t)}T^*_i(0,a-t)\mathds{1}_{\left\{a>\tau\right\}}\in\partial M_i^0.
\end{align*}
 For $t\geq 0$, define $x_i(t)=0$, $x^*_i(t,a)=\frac{\phi_i(a)}{\phi_i(a-t)}T^*_i(0,a-t)\mathds{1}_{\left\{a>\tau+t\right\}}$.  Then, 
 $$\xi(t):=\left(T(t+\tau),V_1(t+\tau),..,x_i(t),..,V_n(t+\tau),T^*_1(t+\tau,a),...,x^*_i(t,a),..,T^*_n(t+\tau,a)\right)$$
  is a solution to the system (\ref{a3}) with initial condition $\xi(0)=S(\tau)x$. Then, by forward uniqueness of solutions, $V_i(t)=0$ and $T^*_i(t,a)\in \partial M_i^0$ for all $t\geq 0$, which is a contradiction to a previous statement..  Thus $\partial X_j$ is forward invariant.  

Now to show $X_j$ is forward invariant.  Notice that $\dot{V_j}\geq -\gamma V_j$.  Hence $V_j(t)\geq V(0)e^{-\gamma_j t}$ for all $t\geq 0$.  If $V_j(0)>0$, then the result follows.  If $V_j(0)=0$, then $\int_0^{\infty}\!p_j(a)T^*_j(0,a)\,da>0$ (since $x(0)\in X_j^0$).  Then $\frac{d}{dt}V_j(0)>0$, so that $\exists \tau>0$ such that $\forall t\in (0,\tau]$, we have $V_j(t)>0$.  Note that in this case, we can choose $\tau$ such that $\int_0^{\infty}\!p_j(a)T^*_j(t,a)\,da>0$ for all $t\in [0,\tau]$.  Then, the same argument applies with $V_j(t)\geq V_j(\tau)e^{-\gamma_j t}$ for $t\geq \tau$.  Hence $V_j(t)>0 \ \forall t>0$.  Then, since $T(t)>0 \ \forall t>0$, we have that $T^*_j(t,a)\geq k_jV_j(t-a)T(t-a)\phi(a)>0$ for all $t>0$.  This shows forward invariance for $j=1$.  For $j>1$, notice that $\partial X_{j-1}$ is forward invariant, which implies forward invariance of $X_j$.  Also, note that $S(t)X_j\subset (X_j)_+$ for all $t>0$ for $j\geq 1$. 

Since $X_j\subset X^0, \ j=1,...,n$, we conclude that $X^0$ is forward invariant.  Also, $\partial X:=\partial X_n$ is forward invariant.  In view of our system and the properties of $f(T)$, it is clear that $\forall x \in \partial X$, we have $S(t)x\rightarrow E_0$ as $t\rightarrow \infty$ where $E_0$ is the infection-free equilibrium. 
\end{proof}

We are now ready to use mathematical induction in order to prove the main result.  The following theorem states that solutions with initial conditions corresponding to positive concentration of $V_1$ or positive productive infected cell density $T_1^*(\cdot)$, will converge to the equilibrium $E_1$ (the single-strain equilibrium belonging to the strain with maximal basic reproduction number).
\begin{theorem} \label{gas}
Suppose that Condition (\ref{repro}) holds, and $f(T)$ satisfies the sector condition (Condition (\ref{(C)})).  Then $E_1$ is globally asymptotically stable for the model (\ref{a3}) with respect to initial conditions satisfying $V_1(0)+\int_0^{\overline{a}_1}T_1^*(0,a)\,da>0$.
\end{theorem}
As mentioned previously, we assume that $\mathcal{R}_n>1$, i.e. Condition (\ref{repro}) holds, in order to simplify the notation.  The case where some of the reproduction numbers are less than unity can easily be adapted to our argument.
\begin{proof}[Proof of Theorem \ref{gas}]
We prove the theorem by induction on the number of strains, $n$, in the system (\ref{a3}). \\
$\bf{n=1:}$   We note that the proof for this case is contained inside following arguments.  Hence, the whole argument can be thought of as self-contained, but this would either make the proof more disorganized or more repetitive.  Therefore, we simply state that the case $n=1$ was proven in Browne and Pilyugin \cite{bropil}. \\ 
\textbf{Induction Step:}  Assume that Theorem \ref{gas} is true for all $n<m$.   We will prove the theorem is true for $n=m$.
\begin{lemma}\label{invariance1}
If $x \in X_j$ where $j\geq 2$, then $S(t)x\rightarrow E_j$ as $t\rightarrow \infty$. 
\end{lemma}
\begin{proof}
 For $j=2,...m$, define the projection operator $P_j:\left(T,V_1,V_2,....,V_m,T^*_1(a),T^*_2(a),.....,T^*_m(a)\right)\mapsto \left(T,V_j,V_{j+1},....,V_m,T^*_j(a),T^*_{j+1}(a),.....,T^*_m(a)\right)$.  Also, define a semigroup of the projected system as follows:  $S_j(t)$ is the semigroup on $\mathbb{R}^{m-j+2}_+\times \prod_1^{m-j+1} L^1_+$ generated by the solutions to the system (\ref{a3}) with $n=m-j+1$ strains, which matches the $m$ strain model except that the first $j-1$ strains are eliminated.  Then $X_j$ is ``projection invariant'' with respect to the system (\ref{a3}), i.e. $P_j(S(t)x)=S_j(t)P_j(x) \ \forall x\in X_j, \ t\geq 0$.  It follows by our induction hypothesis that for any $x\in X_j$, $S_j(t)P_j(x)\rightarrow P_j(E_j)$ as $t\rightarrow \infty$.  Therefore $P_j(S(t)x)\rightarrow P_j(E_j)$ as $t\rightarrow\infty$.  Clearly for $S(t)x\in X_j$, $\left\|T^*_i(t,a)-0\right\|_{L^1}\rightarrow 0$ as $t\rightarrow \infty$ for all $1\leq i\leq j-1$.  Hence, for each $j\geq 2$, $S(t)x\rightarrow E_j \ \forall x\in X_j$. 
 \xqedhere{11.4cm}{\lozenge}
 \renewcommand{\qed}{}
\end{proof}
We continue the proof of the main result by showing uniform persistence and existence of an interior global attractor:
\begin{lemma}\label{lempersist}
The semigroup $S(t)$ is uniformly persistent with respect to $X_1$ and $\partial X_1$.  Moreover, there exists a compact set $\mathcal{A} \subset \left(X_1\right)_+$ which is a global attractor for $ \left\{S(t)\right\}_{t\geq 0}$ in $X_1$, and $\exists \mu>0$ such that
$$ \liminf_{t\rightarrow\infty}V_1(t)\geq \mu, \quad \text{and} \quad \liminf_{t\rightarrow\infty}d(T^*_1(t,a),\partial M^0_1) \geq \mu $$
\end{lemma}
\begin{proof}
We will apply Theorem \ref{semipersist}.  Let $A_{\delta}$ be the strong global attractor of $\partial X_1$.  Also, consider $\widetilde{A}_{\delta}:=\cup_{A_{\delta}}\omega(x)$.  Note that $\partial X_1= \partial X \cup \bigcup_{i=2}^{m} X_i$.  Hence, from Lemma \ref{invariance1}, we obtain that $\widetilde{A}_{\delta}=\left\{E_0,E_2,E_3,...,E_m\right\}$.   We will show that each $\left\{E_i\right\} \subset \widetilde{A}_{\delta} \ i=0,2,3,\dots m$ is an isolated invariant set.  For convenience of notation, suppose $i \in \left\{2,\dots m\right\}$ (the same argument works for $E_0$).  Let $B:=B_{r}(E_i)$ be an open ball of sufficiently small radius $r$ around $E_i$.  We claim that $B$ is an isolating neighborhood.  Suppose by way of contradiction that $\left\{E_i\right\}$ is not a maximal invariant set.  Then, let $M\subset B$ be an invariant set with $M\neq \left\{E_i\right\}$.  There exists a complete orbit $\gamma(x)\subset M$ for $x\in M\setminus \left\{E_i\right\}$.  If $x \in X_i \cup X_{i+1}\cup \dots \cup X_m \cup \partial X$, then $x=E_i$ by Proposition \ref{lyap}, which is a contradiction.  If $x\in X_k$ for $k=2,\dots,i-1$, then $x=E_k$ by Proposition \ref{lyap}, again a contradiction since $E_k\notin B$.   Therefore, $\widetilde{A}_{\delta}$ is isolated.  

To show that $\widetilde{A}_{\delta}$ is acyclic, we need to show that no subset of $\widetilde{A}_{\delta}$ forms a cycle.  Consider $x\in\partial X$.  Then $S(t)x\rightarrow E_0$ as $t\rightarrow\infty$ by the properties of $f(T)$.  Now, suppose that $x\in X_j$ for some $j\geq 2$.  Then $S(t)x\rightarrow E_j$ as $t\rightarrow \infty$ by Lemma \ref{invariance1}.  Hence, $\forall i=2,3,\dots,m$, $x\in W^s(\left\{E_i\right\})\Leftrightarrow x\in X_i$ by Proposition \ref{omega} and the definition of stable manifold.  And $x\in W^s(\left\{E_0\right\}) \Leftrightarrow x\in\partial X$. 

First, let's consider the possibility of a cycle with length greater than or equal to $2$.   This cycle must include a chain with $\left\{E_i\right\} \hookrightarrow \left\{E_j\right\}$ where $i<j$.  For simplicity of notation, consider $2\leq i< j$ ($E_0$ can be handled similarly).  Then, $x\in W^u(\left\{E_i\right\})\cap W^s(\left\{E_j\right\})$ where $i<j$ and $i\in\left\{2,...,m-1\right\}, j\in\left\{3,...,m\right\}$.  Hence, $x\in X_{j}$ for some $j >i$.  Then $V_{i}(0)=0$ and $T^*_{i}(0,a)\in \partial M^0_i$.  The forward invariance of $X_{i}$ requires that $V_{i}(t)=0$ and $T^*_{i}(t,a)\in \partial M^0_i$ for any negative $t$ on a backward orbit through $x$.  Hence, $\alpha(x)\subset \partial X_{i}$, implying that $x\notin W^u\left(\left\{E_{i}\right\}\right)$. This excludes the possibility of cycles of length greater than or equal to $2$ for $S(t)|_{\partial X_1}$.  

Now we consider the possibility of of a $1$-cycle for $S(t)|_{\partial X_1}$.  Then, $\left\{E_i\right\}\hookrightarrow \left\{E_i\right\} $ for some $i=0,2,3,....,m$.  First, we show that that we can not have a $1$-cycle for $E_0$.  It suffices to show that $\left(\partial X\setminus \left\{E_0\right\} \right) \cap W^u(\left\{E_0\right\})=\emptyset$.  Let $x\in \partial X\setminus \left\{E_0\right\}$.  Any backward orbit of $x$ must stay in $\partial X$ since $X^0$ (the complement of $\partial X$) is forward invariant.   If $T^*_i(0,a)=0 \ \forall i=1,...,m$ (in $L^1$), then we have a scalar ODE with a unique positive equilibrium and $\lim_{t\rightarrow -\infty}T(t)=0 \ \text{or} \ \infty$.    The forward invariance of $X^0$ requires $\int_0^{\overline{a}_i}\!T^*_i(t,a)\,da=0$ for any negative $t$ on a backward orbit through $x$.  If $\int_{\overline{a}_i}^{\infty}\!T^*_i(0,a)\,da>0$ for some $i$, then $\int_0^{\overline{a}_i}\!T^*_i(t,a)\,da>0$ for some negative $t$ on a backward orbit through $x$, which is a contradiction.  Therefore, there can be no backward orbit through $x$ if $\int_{\overline{a}_i}^{\infty}\!T^*_i(0,a)\,da>0$ for some $i$.  Hence, $E_0$ cannot be an $\alpha$-limit point of any $x\in \partial X\setminus\left\{E_0\right\}$.  Now consider the case $x\in X_j$ for some $j\geq 2$.  Suppose by way of contradiction that $\left\{E_j\right\}\hookrightarrow\left\{E_j\right\}$.  Thus, $x\in \left(W^s(\left\{E_j\right\})\cap W^u(\left\{E_j\right\})\right)\setminus\left\{E_j\right\}$.  Then there exists a complete orbit $z(t)$ through $x$, such that $z(t)\rightarrow E_j$ as $t\rightarrow \pm \infty$.  Here, $z(t)$ is a homoclinic orbit.  Notice that the positive invariance of $\left(X_j\right)_+$ implies that $V_i(t)>0$ for all $t\in\mathbb{R}$.  The continuity and positivity of $V_j$, along with the fact that $\lim_{t\rightarrow\pm\infty}V_j(t)=\overline{V}_j^j>0$, imply that there exists $\epsilon>0$ such that $V_j(t)\geq\epsilon \ \forall t\in\mathbb{R}$.  In a similar fashion, we can show that $T(t)\geq\epsilon$ for all $t\in\mathbb{R}$.  Note that both $X_j$ and $X\setminus X_j$ are forward invariant.  So, $z(t)\in X_j \ \forall t\in\mathbb{R}$.  $X_j$ is projection invariant with respect to $P_j$ and $S_j$ as defined earlier.  In other words, in $X_j$, we can consider an equivalent $n=m-j+1$ strain model with $\mathcal{R}_j$ as the maximal reproduction number.  In this case, Proposition \ref{lyap} applies to $P_j(z(t))$ and semigroup $S_j(t)$.  We can conclude that $x=E_j$, which is a contradiction.  Hence $\widetilde{A}_{\delta}$ is acyclic for $S(t)|_{\partial X_1}$.  

To finish the proof of uniform persistence, we need to prove:
$$W^s(\left\{E_i\right\})\cap X_1 \ \  \ i=0,2,3,....,m$$
Suppose by way of contradiction that there exists $x\in X_1$ such that $x\in W^s(\left\{E_i\right\})$ where $i=2,...,m$ (the following argument can also be applied for $E_0$).  By Proposition \ref{omega}, $S(t)x\rightarrow E_i$.  By the semigroup property, we can find a sequence $(x_{\ell})\subset X_1$ such that
$$ \left\|S(t)x_{\ell}-E_i\right\| < \frac{1}{\ell} \quad \forall t\geq 0. $$
Let $S(t)x_{\ell}=(T^\ell(t),V_1^\ell(t),...,\left(T^*_1\right)^\ell(t,a),....)$ and $x_{\ell}=(T^\ell(0),V_1^\ell(0),....,\left(T^*_1\right)^\ell(0,a),....)$.  Then we have
$$|T^\ell(t)-\overline{T}_i|\leq \frac{1}{\ell}, \quad \forall t\geq 0.$$
Then by applying a simple comparison principle, we deduce that  $V_1^\ell(t)\geq y_1^\ell(t)$  where $y_1^\ell(t)$ is a solution of
\begin{align*}
 \frac{dy_1^\ell(t)}{dt} = \int_0^{t} \!k_1p_1(a)\phi_1(a)(\overline{T}_i-\frac{1}{\ell})y_1^\ell(t-a)\,da - \gamma_1 y_1^\ell(t) , \quad  y_1^\ell(0)  =V_1^\ell(0). 
\end{align*}
Note that if $V_1^\ell(0)=0$, then clearly $\left(T^*_1\right)^\ell(0,a)\in M^0$ and hence $\dot{y}_{\ell}(0)>0$, so without loss of generality we can take $V_1^\ell(0)>0$.   We claim that for $n$ sufficiently large, $y_1^\ell$ is unbounded.  The assumption $\overline{T}_1<\overline{T}_i$ is equivalent to $-\gamma_1 + k_1\overline{T}_i\int_0^{\infty}\! p_1(a)\phi_1(a)\,da>0$.  Hence $\exists N\in\mathbb{N}$ such that $-\gamma_1 +  k_1\left(\overline{T}_i-\frac{1}{N}\right)\int_0^{\infty}\! p_1(a)\phi_1(a)\,da>0$.  By Lemma 3.5 in \cite{bropil}, we conclude that $y_1^N$ is unbounded.   Since $V_1^N\geq y_1^N$, we obtain that $V_1^N$ is unbounded and hence $S(t)x_N$ is unbounded which is certainly a contradiction.  Therefore, $W^s(\left\{E_i\right\})\cap X_1 = \emptyset$.  By Theorem \ref{semipersist}, we find that $S(t)$ is uniformly persistent with respect to $X_1$ and $\partial X_1$, i.e. $\partial X_1$ is uniform strong repeller.  
Then, by Theorem \ref{persistattract}, we can conclude that there exists a compact set $\mathcal{A} \subset X_1$ which is a global attractor for $ \left\{S(t)\right\}_{t\geq 0}$ in $X_1$.  Since $S(t)X_1\subset \left(X_1\right)_+$, the global attractor, $\mathcal{A}$, is actually contained in $\left(X_1\right)_+$.  Because of this, $\exists \mu>0$ such that
$$ \liminf_{t\rightarrow\infty}V_1(t)\geq \mu, \quad \text{and} \quad \liminf_{t\rightarrow\infty}d(T^*_1(t,a),\partial M^0_1) \geq \mu \xqedhere{3.5cm}{\lozenge}$$
 \renewcommand{\qed}{}
\end{proof}
\vspace{.4cm}
Because the interior global attractor $\mathcal{A}$ is invariant, we can find a complete orbit through any point contained in $\mathcal{A}$.  For any complete orbit $\left\{z(t):t\in\mathbb{R}\right\}\subset \mathcal{A}$, there exists $\epsilon>0$ such that $V_1(t)\geq\epsilon$ and $T(t)\geq\epsilon$ for all $t\in\mathbb{R}$.  Hence, by Proposition \ref{lyap}, $\mathcal{A}=\left\{E_1\right\}$.  Thus $\left\{E_1\right\}$ is the globally attractor for the model (\ref{a3}) with respect to initial conditions satisfying $V_1(0)+\int_0^{\overline{a}_1}T_1^*(0,a)\,da>0$.  A global attractor is also locally stable by definition, therefore $E_1$ is indeed globally asymptotically stable.
\end{proof}
Hence, we have proved that the viral strain with maximal reproduction number competitively excludes all of the over viral strains.

\section{Simulations and Application to HIV Evolution}\label{s6}
The purpose of this section is to numerically illustrate the main result, Theorem \ref{gas}, and also to investigate the transient dynamics.  For applications, steady state behavior is not the only important consideration since the rate at which the equilibrium is achieved can give information on the evolution of the virus.  For example, the persistence of HIV in a patient is dependent on its ability to evolve resistance to specific immune pressures and the rate of this evolution can provide insights into the patient's immune system and disease progression \cite{althaus}.

We consider two scenarios:  first, a case where two strains are present at low numbers in a wholly susceptible target cell population, and second, the case where one strain is at steady state and a strain with larger reproduction number is introduced into the system.  From Theorem \ref{gas} we know that, asymptotically, the strain with larger reproduction number will competitively exclude inferior strains, but to learn about transient dynamics and the rate at which a strain is replaced, we need to have an idea about the rate at which the virus strains undergo their replication cycle.

One method of formulating the ``replication speed'' of a strain $i$ is to calculate the viral generation time (which we denote by $G_i$), as defined by Perelson and Nelson \cite{stmodel} for the single strain ODE model.  The method assumes that the system is at a single-strain equilibrium $T=\overline{T}_i, V_i=\overline{V}_i$ at $t=0$ and keeps track of new virus particles, $v_{{\rm new}}$, created by the initial virus particles.  To do this, consider the equations:
\begin{align*}
\frac{dV_i}{dt}&= -\gamma_i V_i, \qquad V_i(0)=\overline{V}_i \\
\frac{dv_{{\rm new}}}{dt}&= \int_0^t p_i(a) \phi_i(a) k_i V_i(t-a) \overline{T}_i \,da
\end{align*}
Define the cumulative probability distribution of producing a virion by time $t$ as $P(t)=\frac{v_{{\rm new}}(t)}{\overline{V}_i}$.  Then, the average time of virion production is given by
\begin{align*}
G_i= \int_0^{\infty} t \frac{dP(t)}{dt} \,dt =\frac{1}{\overline{V}_i} \int_0^{\infty} t \int_0^t p_i(a) \phi_i(a) k_i V_i(t-a) \overline{T}_i \,da \,dt
\end{align*}
Inserting $V_i(t-a)=e^{-\gamma_i(t-a)}$, switching the order of integration and integrating, we obtain 
\begin{align*}
G_i=\frac{\int_0^{\infty}a\phi_i(a)p_i(a)\,da}{\int_0^{\infty}\phi_i(a)p_i(a)\,da}+\frac{1}{\gamma_i} \tag{16}\label{gen}
\end{align*}
Notice that $G_i$ can be interpreted as the average age of viral production divided by the average number of virus produced by infected cell, plus the average lifespan of free virus particle.

While the viral generation time, $G_i$, can be presented in a nice formula (\ref{gen}) and $\mathcal R_i/G_i$ gives an idea for a value of ``replication rate'', perhaps a more accurate descriptor is the virus growth rate, $\lambda_{i,j}$, of the linearized system for $V_i$ at the equilibrium $E_j$ ($j\neq i$).   Consider the linearized equation,
\begin{align*}
\frac{dV_i}{dt}&= \int_0^{\infty} p_i(a) \phi_i(a) k_i V_i(t-a) \overline{T}_j \,da - \gamma_i V_i.
\end{align*}
The exponential growth rate for $V_i$ in this linearized equation is the principal eigenvalue, $\lambda_{i,j}$, where $V_i(t)=e^{\lambda_{i,j} t}$.  Thus, $\lambda_{i,j}$ satisfies the equation $$\lambda_{i,j}+\gamma_i=k_i\overline{T}_0\int_0^{\infty} p_i(a) \phi_i(a) e^{-\lambda_{i,j} a} \,da .$$
It is not hard to see that there is a unique eigenvalue, $\lambda_{i,j}$, satisfying the above equation, and if $j=0$ then $\lambda_{i,j}>0 \Leftrightarrow \mathcal R_i>1$; if $j>0$ then $\lambda_{i,j}>0 \Leftrightarrow \mathcal R_i> \mathcal R_j$.

\begin{figure}[t]
\subfigure[][]{\label{fig1a}\includegraphics[width=8.5cm,height=6cm]{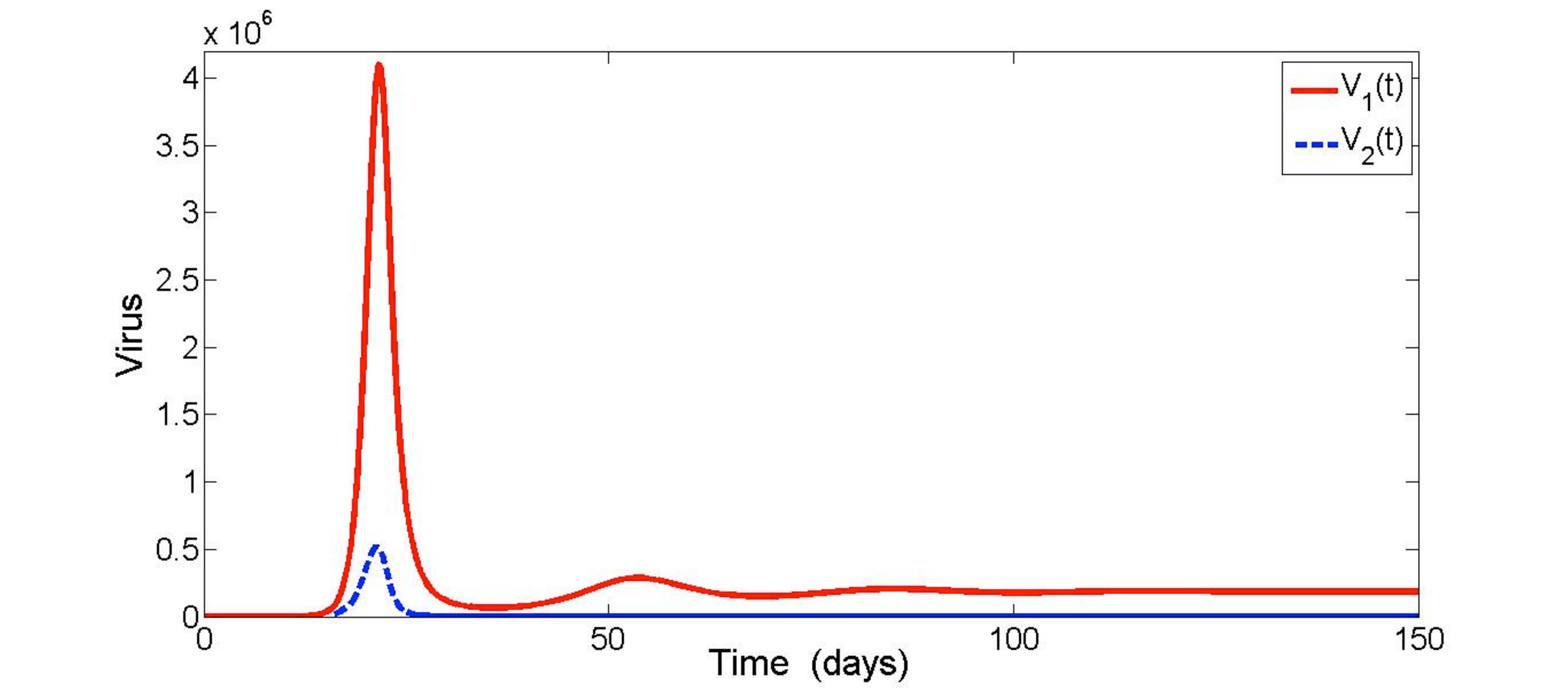}}
\subfigure[][]{\label{fig1b}\includegraphics[width=8.5cm,height=6cm]{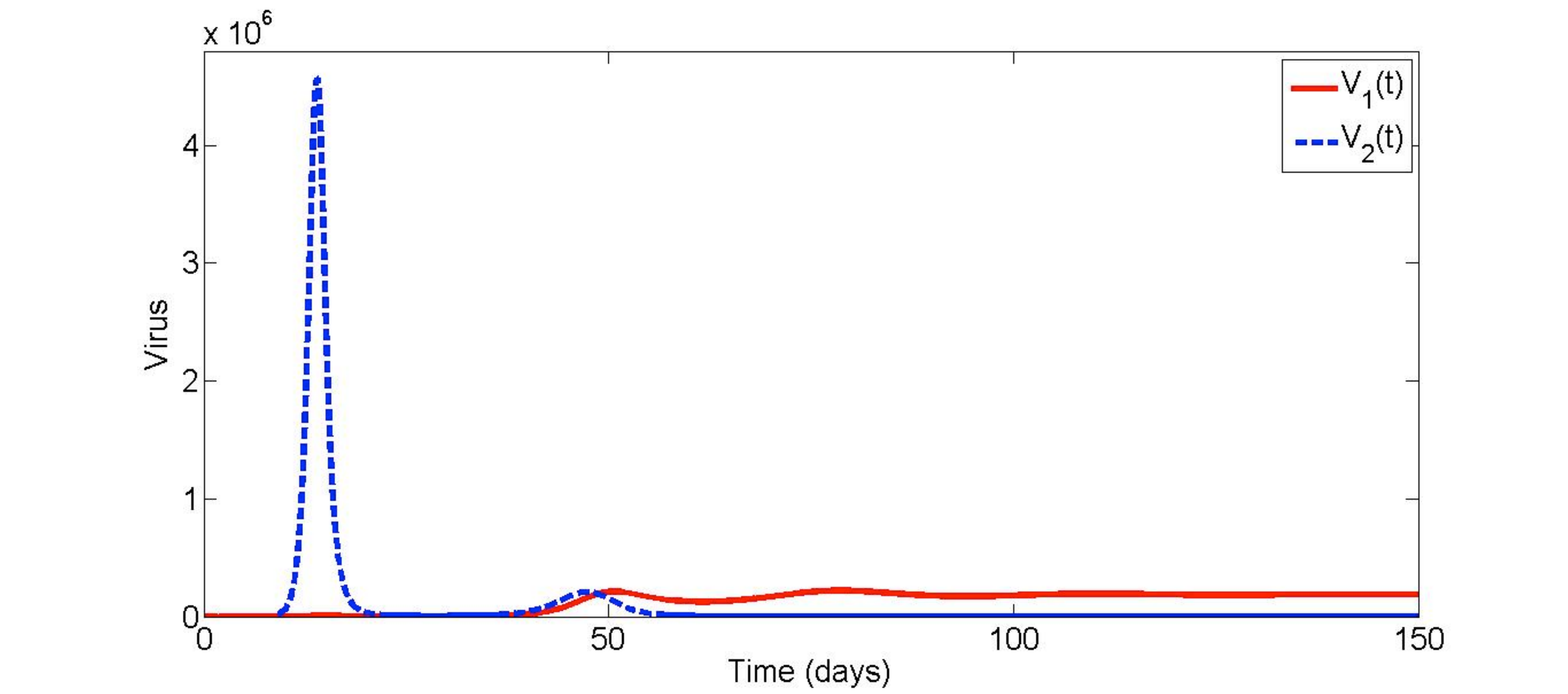}}
\caption{\emph{Dynamics for initial infection with two strains $V_1,V_2$ with $\mathcal R_1>\mathcal R_2$}. (a)  Parameters are chosen so that replicative speed of strain 1 is greater than strain 2 ($\lambda_{1,0}>\lambda_{2,0}$).  $V_1$ (solid line) dominates for the entire time.  (b)  Parameters are chosen so that replicative speed of strain 1 is less than strain 2 ($\lambda_{1,0}<\lambda_{2,0}$).  $V_2$ (dashed line) dominates early times before being competitively excluded by $V_1$. }
  \label{fig1}
  \end{figure}

In the simulations, we consider a linear healthy cell net growth rate $f(T)=s-cT$ and two virus strains $V_1,V_2$ with infected cell death rates $\delta_i(a)$ and viral production rates $p_i(a)$ of the following piecewise form:
\begin{align*}
\delta_i(a)&=\begin{cases} \mu_i & 0\leq a < \tau_i \\  \nu_i & \tau_i <a \end{cases}, \qquad p_i(a)=\begin{cases} 0 & 0\leq a < \tau_i \\  \rho_i & \tau_i <a \end{cases}  \qquad  \text{ for } i=1,2.
\end{align*}
Here $\tau_i$ is the intracellular delay between cell infection and viral production.  Note that by defining $I_i(t)=\int_{\tau_i}^{\infty}T^*_i(t,a)\,da $, system (\ref{a3}) can reduce to the following delay differential equation:
\begin{align*}
\frac{dT(t)}{dt}&=s-cT-\sum_i k_iV_i(t)T(t) \\
\frac{dI_i(t)}{dt}&= e^{-\mu_i\tau_i}k_iV_i(t-\tau_i)T(t-\tau_i)-\nu_iI_i (t)\\
\frac{dV_i}{dt}&=p_i I_i(t) -\gamma_i V_i(t)
\end{align*}

We assume that strain 1 has the largest reproduction number in the simulations, with the following parameters for strain 1 and the target cells:  $s=10^4 \ {\rm ml}\times{\rm day}^{-1}$, $c=0.01 \ {\rm day}^{-1}$, $k_1= 8\times 10^{-7} \ {\rm ml}^{-1}\times {\rm day}^{-1}$, $\gamma_1=13 \ {\rm day}^{-1}$, $\tau_1=2 \ {\rm day}$, $\mu_1=0.05 \ {\rm day}^{-1}$, $\nu_1=0.7 \ {\rm day}^{-1}$, and $p_1$ will be varied.  We note the parameters are within the range of suitable choices for HIV infection \cite{delaymod}.

In the first scenario, we consider the case where the two virus strains are introduced into a healthy target cell population at low density.  Hence, we assume that $T(0)=\overline{T}_0=s/c, T^*_i(0,a)\equiv 0, V_i(0)=1$ for $i=1,2$.   In Figure \ref{fig1a}, all parameters for strain 2 are identical to that of strain 1, except $\nu_2=1.14$.  This change in parameters result in a slower replicative speed for $V_2$ and lower reproduction number in comparison with $V_1$; namely $\mathcal R_2=9.7562$, $G_2=2.95$, $\mathcal R_2/G_2=3.3$, and $\lambda_{2,0}=0.8337$.  It is seen that $V_1$ dominates from the initial infection to the competitive exclusion.  In contrast, if we choose parameters where replicative speed for $V_2$ is faster than that of $V_1$ (but $\mathcal R_2$ is the same as in Figure \ref{fig1a}), then the initial peak is dominated by $V_2$ before $V_1$ competitively excludes $V_2$ as seen in Figure \ref{fig1b}.  Here $p_1=p_2=200 {\rm day}^{-1}$.  The parameters result in a reproduction number $\mathcal R_1=15.9$, viral generation time $G_1=3.5$, $\mathcal R_1/G_1=4.5$, and initial growth rate at infection-free equilibrium of $\lambda_{1,0}=0.93$.  The parameters for $V_2$ different from $V_1$ in Figure \ref{fig1b} are $\tau_2=1$ and $\nu_2=1.2$, resulting in $\mathcal R_2=9.7562$, $G_2=1.91$, $\mathcal R_2/G_2=5.1$, and $\lambda_{2,0}=1.4016$.  Thus, we can speculate that the initial peak of viral load in HIV may be dominated by strains with high replicative speed, but they may taken over by strains with lower replicative speed but higher reproduction number, as considered for an ODE mutation model in \cite{ball}.

  \begin{figure}[t]
\subfigure[][]{\label{fig2a}\includegraphics[width=8.5cm,height=3.5cm]{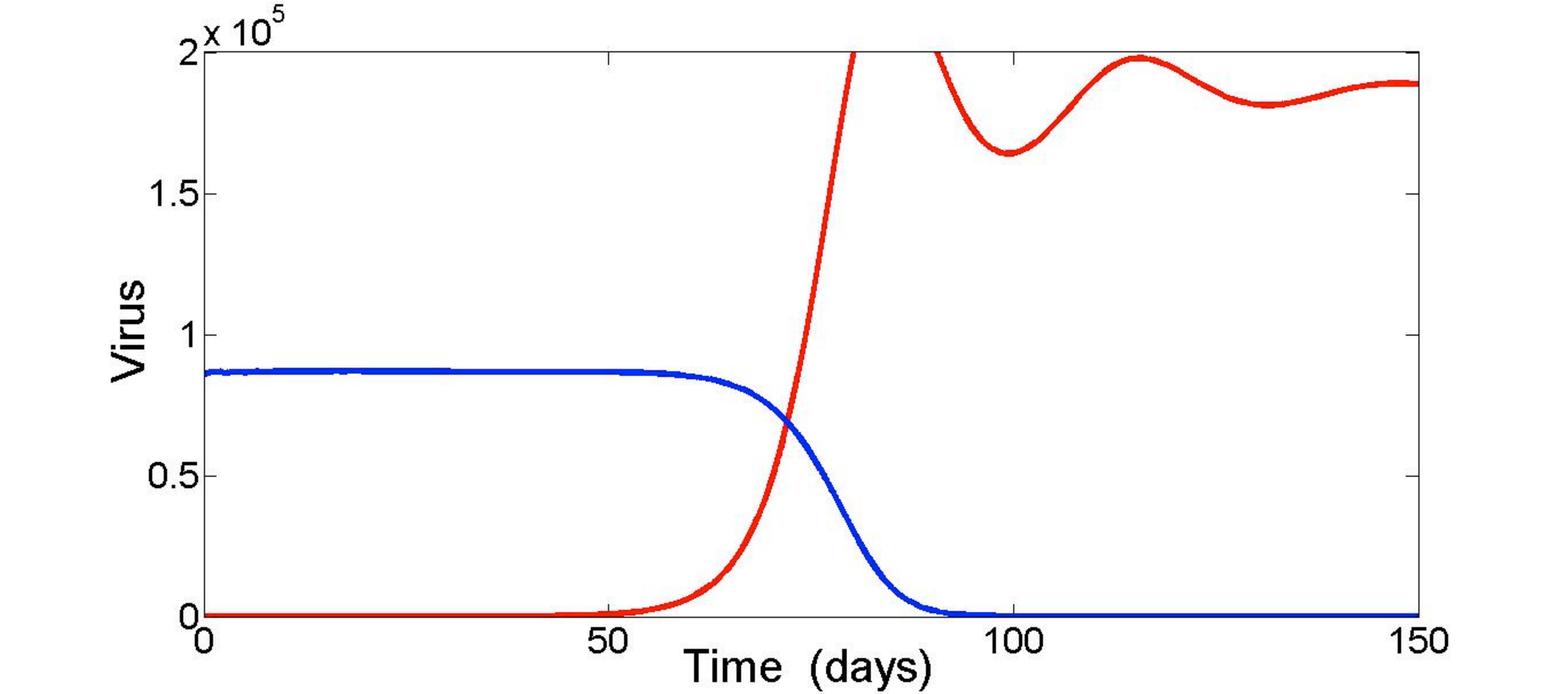}}
\subfigure[][]{\label{fig2b}\includegraphics[width=8.5cm,height=3.5cm]{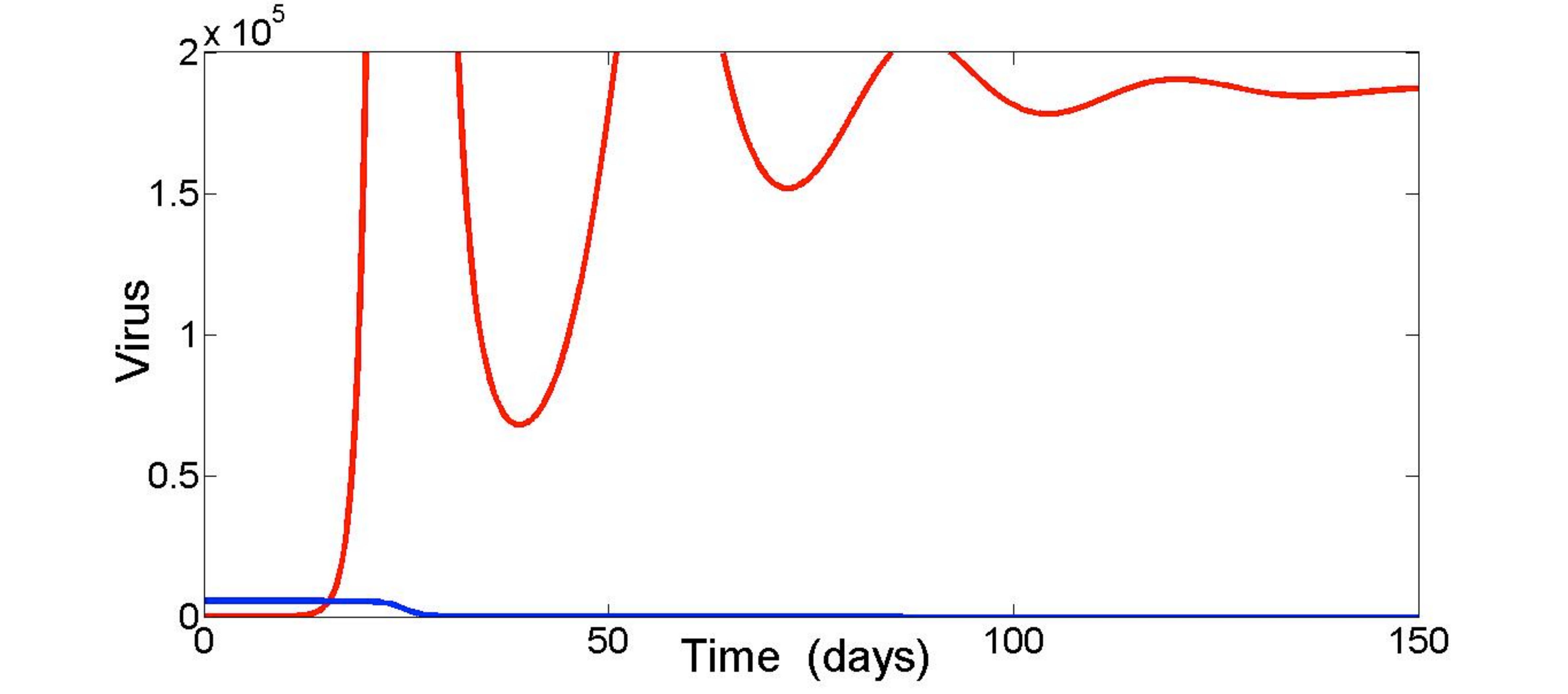}} \\
\subfigure[][]{\label{fig2a}\includegraphics[width=8.5cm,height=3.5cm]{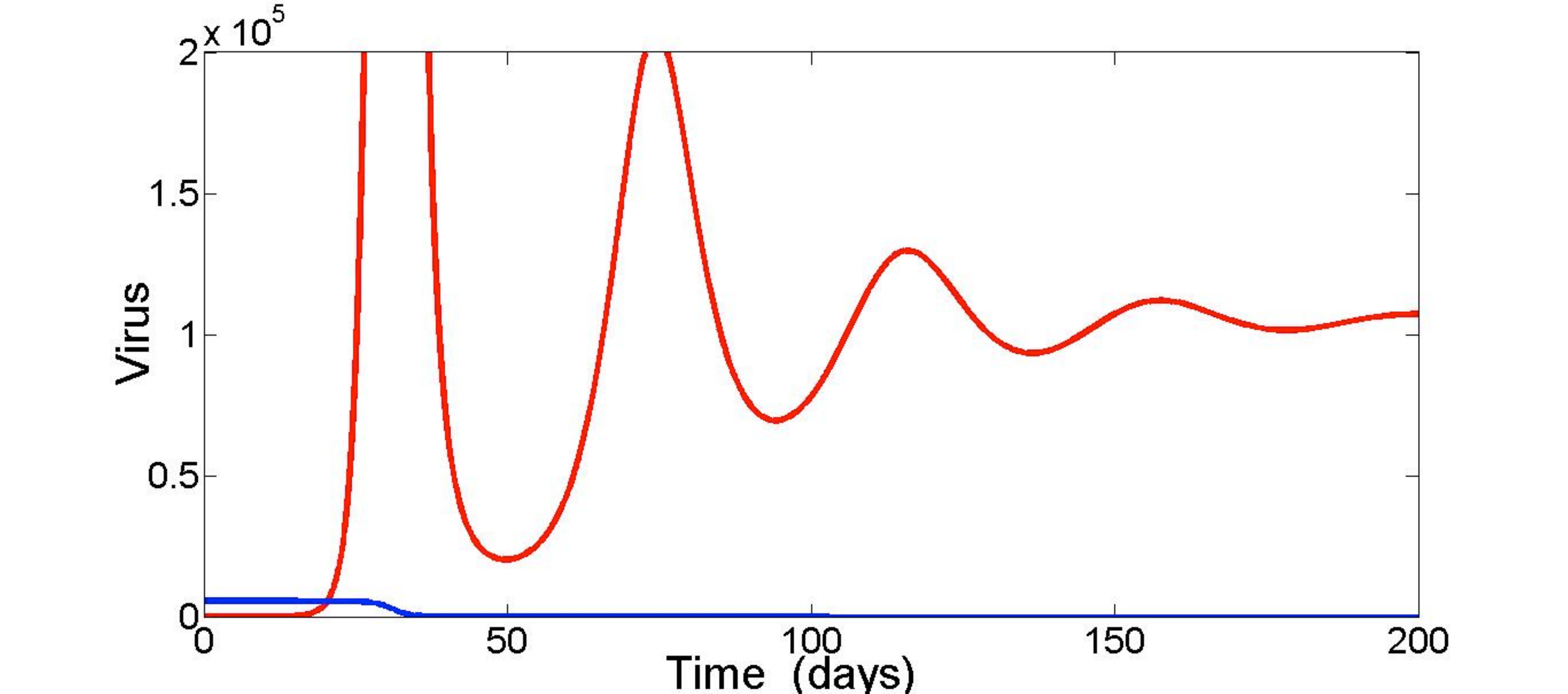}}
\subfigure[][]{\label{fig2b}\includegraphics[width=8.5cm,height=3.5cm]{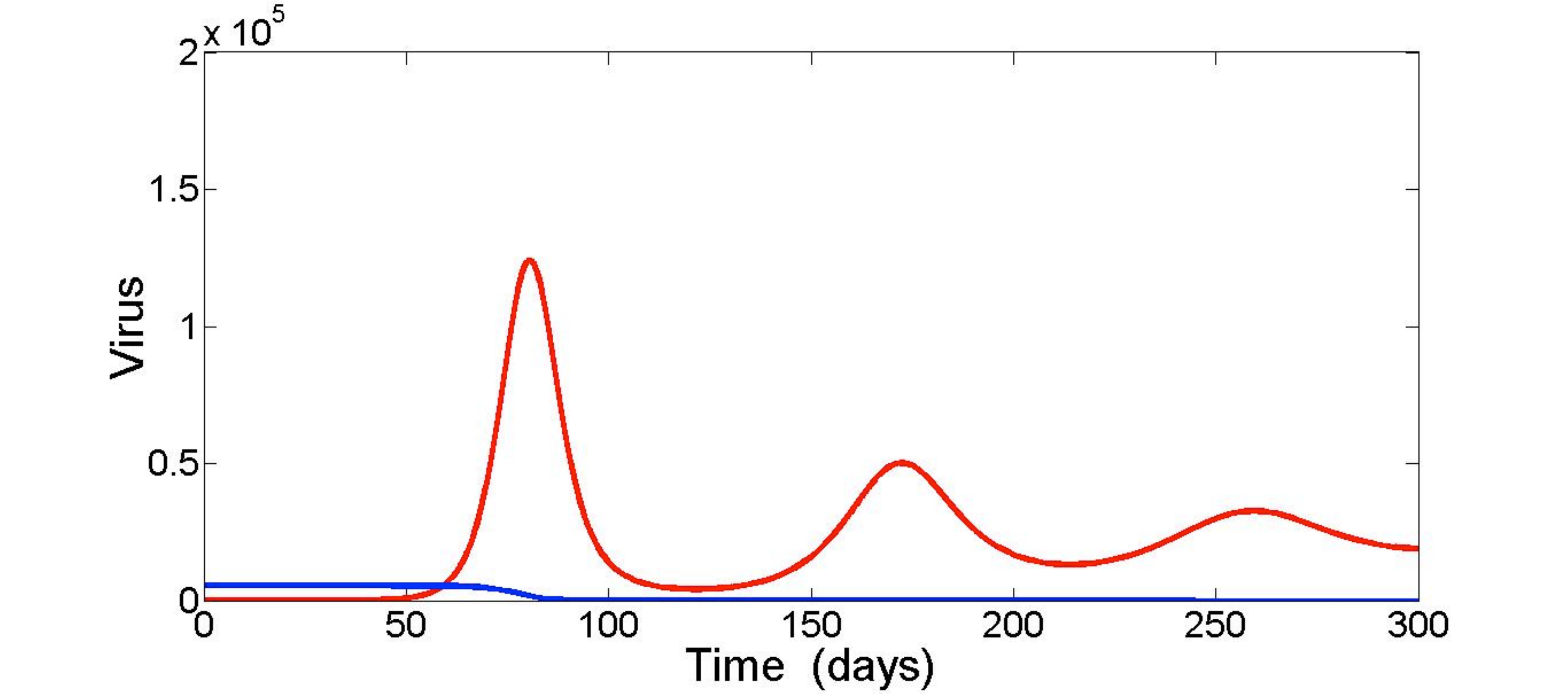}} 
\caption{ \emph{Comparison of strain replacement after CTL attack with early killing versus late killing and different fitness costs for the virus.}  In these simulations, an immune response attacks an epitope presented by $T^*_2$, imposing constant death rate $d$ either before or after viral production starts.  $V_1$ (red) and $V_2$ (blue) are plotted as a function of time. The simulations start at the $V_2$ steady state (after immune attack), and an escape mutant $V_1$ is introduced.  In (a), the death rate $d$ is imposed on $T^*_2$ after viral production starts (late killing).  In (b), (c), and (d), the death rate $d$ is imposed on $T^*_2$ before viral production (early killing).  The reproduction numbers (i.e. fitness) of the strains are as follows:  (a) $\mathcal R_1=15.9, \mathcal R_2=7.9$; (b) $\mathcal R_1=15.9, \mathcal R_2=1.4$; (c) $\mathcal R_1=8.4, \mathcal R_2=1.4$; (d) $\mathcal R_1=2.8, \mathcal R_2=1.4$.  Early killing suppresses $V_2$ to a lower level than late killing, but it is necessary that there is high fitness cost of the mutant for slower immune escape and lower viral load of $V_1$. }
  \label{fig1}
  \end{figure}

In the second scenario, we investigate strain replacement.  Hence, we assume that $V_2$ is at steady state and introduce $V_1$ into the system.  A motivation for this scenario is HIV immune escape, where the virus evolves resistance to attack from the immune responders cytotoxic T lymphocytes (CTLs) \cite{Vitaly, althaus}.  There has been considerable interest in quantifying rates at which escape variants replace a previous virus strain \cite{Vitaly}.  Also, there is recent evidence that different CTL clones respond to epitopes presented on the infected cell at different stages in the infected cell life cycle, for example before viral production or after initiation of viral production \cite{sacha,Kloverpris}.  We consider the scenario where a constant (non-explicit) immune response attacks the dominant virus strain, labeled $V_2$, with killing rate $d$ against an epitope presented either before or after viral production, and a escape mutant, $V_1$, replaces strain $V_2$ in Figure \ref{fig1}.  Thus, $\delta_2(a)=\delta_1(a)+d\mathds{1}_{\left\{t<\tau\right\}}$ or $\delta_2(a)=\delta_1(a)+d\mathds{1}_{\left\{t>\tau\right\}}$ for ``early killing'' or ``late killing'', respectively, where $\tau=\tau_1=\tau_2=2$.  When $V_2$ reaches the single-strain steady state with this new death-rate, the mutant immune-resistant virus, $V_1$, is introduced into the system without the additional death rate $d$, but with a fitness cost in the virion production, i.e. $p_1(a)=cp_2(a)$ where $c<1$.  We find that for a given killing rate $d$, early killing is substantially more efficient than late killing since it suppresses the $V_2$ population to a much lower steady state.  This aligns with the experimental results in \cite{Kloverpris}.  Also, observe that the efficient early killing applies a larger selection pressure and the immune escape is much more rapid in the case of early killing, assuming that the ``fitness cost'' $c$ is the same for each case.  However, if we assume the ``fitness cost'' is larger, i.e. $c$ is smaller, then the escape will be less rapid and the steady state of $V_1$ will be reduced.  Thus, this analysis suggests a characteristic for successful immune response and reduced viral load may be early killing on a conserved epitope.  We note that the relative rates of strain replacement seen in the simulations can be inferred by comparing the values of the ``invasion'' growth rate $\lambda_{1,2}$ for the different parameters (not shown).  In future work, we will conduct deeper investigation of modeling CTL attack at different stages of the infected cell life cycle and the resulting dynamics.

\section{Discussion}\label{discuss}
Multi-strain models have received much attention in both between-host and within-host disease modeling.  A primary objective has been to determine when the competitive exclusion principle holds versus when coexistence of pathogens can occur.  In a classic result of mathematical epidemiology, Bremerman and Thieme proved competitive exclusion along with the principle of $\mathcal{R}_0$ maximization for an $SIR$ multi-strain model \cite{breb}.  Mechanisms for coexistence of multiple strains in epidemiological models include partial cross-immunity \cite{cross}, superinfection \cite{supinfect}, co-infection \cite{coinfect}, density dependent host mortality \cite{host}, and host population structure \cite{chavez}.  For within-host models, the competitive exclusion and $\mathcal{R}_0$ maximization principle have been proved for the standard virus model \cite{deleenheerandpilyugin} and a stage-structured within-host malaria model \cite{malaria}.  Coexistence of multiple strains in a within-host virus model can occur when immune response is explicitly included, as shown by Souza \cite{souza} in the case of strain-specific immune response.

In this paper, we analyzed a multi-strain within-host virus model with continuous infection-age structure in the infected cell compartment.  The main result is global convergence to the single strain equilibrium of the virus strain which maximizes the basic reproduction number.  In other words, both the competitive exclusion principle and the principle of $\mathcal{R}_0$ maximization holds.  

McCluskey and others have recently found global stability results for a few continuous age-structured models (among these models is the single-strain version of the model (\ref{a1}) analyzed by Browne and Pilyugin) \cite{bropil, magal,mcc}.  The general strategy has been to formalize the problem in terms of semigroup theory, show existence of an interior global attractor, and then define a Lyapunov functional on this attractor.  To show existence of the interior global attractor, uniform persistence must be proved, and hence, the boundary flow must be characterized.  For our multi-strain virus model (model (\ref{a1})) with $m$ strains and $m$ single-strain equilibria, nested inside the appropriate boundary set is an $n$-strain sub-model for all $n<m$.  Hence, the situation calls for strong mathematical induction to be utilized with the induction hypothesis of global asymptotic stability.  After applying the induction argument and checking other conditions, we can establish uniform persistence and then, via a Lyapunov functional, we prove global attractiveness of the single-strain equilibrium belonging to the strain with maximal reproduction number.  

Finally, we simulated the dynamics of the model (\ref{a1}) for specific examples relevant to HIV evolution.  In addition to demonstrating the main result of competitive exclusion, the simulations allowed us to gain insight and explore some formulas for the rate of viral evolution.  From a broader perspective, there are many factors to consider in the evolution of a virus.  Co-evolution with hosts, between-host epidemiological dynamics, within-host competition for target cells and evasion from immune response, application of drug treatment or vaccines, and bio-chemical limitations on replication speed and accuracy, all shape the evolution of viruses \cite{virulence, chem}.   Future work will entail investigating how various factors affect the evolution of viral strains, along with characterizing the within-host dynamics.

\bigskip
{\bf Acknowledgments.}  CJB thanks the anonymous reviewers for their valuable comments and suggestions, along with Professor Sergei Pilyugin for interesting discussions and valuable comments.

\def\bibindent{1.6em}

\end{document}